\documentclass[a4paper]{amsart}

\usepackage[utf8]{inputenc}
\usepackage{amsmath}
\usepackage{amsthm}
\usepackage{amsfonts}

\usepackage{graphicx}
\usepackage{color}
\graphicspath{{pic/}}

\usepackage[pdftex]{hyperref}

\title[Jang's equation and marginal surfaces]{Jang's equation and its applications to marginally trapped surfaces}

\author[L. Andersson]{Lars Andersson}%${}^\dagger$} 
\email{laan@aei.mpg.de}
\address{Albert Einstein Institute, Am M\"uhlenberg 1, D-14476 Potsdam,
  Germany}

\author[M. Eichmair]{Michael Eichmair${}^\dagger$}
\email{eichmair@math.mit.edu}
\thanks{${}^\dagger$ Research partially supported by Australian Research Council
    Discovery Grant \#DP0987650 and by the NSF grant DMS-0906038.}
\address{Department of Mathematics, MIT, 77 
Massachusetts Avenue, Cambridge, MA 02139, USA}
\author[J. Metzger]{Jan Metzger} 
\email{jan.metzger@uni-potsdam.de}
\address{Universit\"at Potsdam, Institut für Mathematik, Am Neuen
  Palais 10, 14469 Potsdam, Germany}

\date{June 23, 2010}

\newcommand{\mnote}[1]{}        %no Marginal Note
\newcommand \R      {\operatorname{R}}

\newcommand \inj    {\operatorname{inj}}

\newcommand \tr     {\operatorname{tr}}
\renewcommand \div  {\operatorname{div}}
\newcommand \dist   {\operatorname{dist}}

\newcommand \mc     {\operatorname{H}}

\renewcommand \L      {\operatorname{L}}
\newcommand \Rad      {\operatorname{Rad}}

\newcommand \vol {\operatorname {vol}}
\newcommand{\IR}{\mathbb{R}}
\renewcommand{\div}{\operatorname{div}}
\newcommand{\la}{\langle}
\newcommand{\ra}{\rangle}
\newcommand{\del}{\partial}
\newcommand{\Vol}{\operatorname{Vol}}
\newcommand{\Area}{\operatorname{Area}}
\newcommand {\Ric}  {\operatorname{Ric}}
\newcommand {\Scal}  {\operatorname{R}}
\newcommand {\graph}{\operatorname{graph}}
\newcommand \T {\operatorname {T}}
\newcommand{\eps}{\varepsilon}

\newcommand{\JM}{{\hat M}}
\newcommand{\Jg}{{\hat g}}
\newcommand{\Jk}{{\hat k}}
\newcommand{\Jnu}{{\hat\nu}}

\newcommand{\CT}{\mathcal{T}}

\newtheorem {theorem}    {Theorem}    [section]
\newtheorem {lemma}      {Lemma}      [section]
\newtheorem {definition} {Definition}   [section]

\newtheorem* {conjecture*}       {Conjecture}
\newtheorem* {theorem*}          {Theorem}

\newtheorem* {acknowledgements*} {Acknowledgements}

\theoremstyle {definition} 
\theoremstyle {definition}

\begin{document}
\begin{abstract} 
  In this paper we survey some recent advances in the analysis of
  marginally outer trapped surfaces (MOTS). We begin with a systematic
  review of results by Schoen and Yau on Jang's equation
  and its relationship with MOTS. We then explain recent work on
  the existence, regularity, and properties of MOTS and discuss the
  consequences for the trapped region. We include an outlook with some
  directions for future research.
\end{abstract}

\maketitle

\tableofcontents

\section{Introduction}
\label{sec:introduction}
Given a Riemannian $3$-manifold $(M,g)$ and a symmetric $(0, 2)$-tensor $k$ on $M$,
the triple $(M,g,k)$ is an initial data set for flat Minkowski spacetime if
and only if the overdetermined 
system of equations
\begin{equation}\label{eq:jang-thm1}  
\begin{split} 
g_{ab} &= g^{\text{flat}}_{ab} - D_a u D_b u \\
k^{ab} &= \frac{D^a D^b u}{(1+D^c u D_c u)^{1/2}} 
\end{split} 
\end{equation} 
has a solution for some flat metric $g^{\text{flat}}$. 
Here indices are raised using the
metric $g_{ab}$. 

This statement and its proof appear in the paper \cite{Jang:1978} by Pong-Soo
Jang, who attributes it to Robert Geroch. 
The details of the calculation leading to \eqref{eq:jang-thm1} can be found in
\cite[Appendix]{Jang:1978}, see also \cite[Appendix C]{Bray-Khuri:2009A}. 
In his paper Jang 
sets out to generalize Geroch's approach to proving the positive mass theorem
(based on the inverse mean curvature flow) from the case of time-symmetric
initial data to the general case.  

Recall that the Geroch mass
 for a two-surface $\Sigma \subset M$ with scalar curvature $\R_\Sigma$ 
and mean curvature $\mc$ is defined
by 
\begin{figure}[!hbt]
\centering 
\resizebox{!}{2in}{\input{pic/jang-lor}} 
\caption{}
\end{figure}
\begin{equation}
16 \pi m_{\text{Geroch}}(\Sigma) = \sqrt{\frac{|\Sigma|}{16\pi}} 
 \int_\Sigma (2\R_\Sigma - \mc_\Sigma^2 )d\mu_\Sigma. 
\end{equation}
The Geroch mass is a specialization of
the Hawking mass \cite{hawking:1968} to the time-symmetric case. The 
explicit form given here appeared in \cite{jang:wald:1977}.
Although phrased slightly differently in \cite{Geroch:1973}, 
Geroch's argument for the positive mass theorem is 
based on the observation that $m_{\text{Geroch}}(\Sigma_s)$ is monotone
increasing 
for a smooth family $\Sigma_s$ moving in the normal direction with speed
given by the inverse mean curvature, in a $3$-manifold of non-negative scalar
curvature. This \emph{inverse mean curvature flow} 
was later analyzed by Huisken and Illmanen \cite{Huisken-Ilmanen:2001}, 
who were able to prove
monotonicity of the Geroch mass for a  weak
version of the flow, and use this to  give a proof of the Riemannian Penrose inequality. 

In generalizing Geroch's argument, Jang introduces the defects
\begin{equation}
  \label{eq:2-intro}
  \begin{split}
    \Jg_{ab}(u)
    &=
    g_{ab} + D_a u D_b u  
    \\    
    \Jk_{ab}(u)
    &=
    k_{ab} - \frac{D_a D_b u}{(1 + D^c u D_c u)^{1/2}} .
  \end{split}
\end{equation}
The condition that $(M,g,k)$ forms initial data for Minkowski space is
therefore equivalent to the condition that $\Jk(u) = 0$ and $\Jg(u)$
be flat for some function $u$.  Taking the trace of $\Jk(u)$ with
respect to the metric $\Jg(u)$ yields the quasilinear equation
\begin{equation}
  \label{eq:jang}
  \left(g^{ab} - \frac{D^a u D^b u}{1 + D^c u D_c u}\right)
  \left(k_{ab} - \frac{D_a D_b u}{\sqrt{1 + D^c u D_c u}}\right) = 0,
\end{equation}
which in particular must be satisfied for the height function 
of any spacelike
asympotically flat hypersurface in Minkowski space. 
This is Jang's equation. 

At this point it is convenient to note that $\Jg_{ab}(u)$ is precisely the
metric induced on the graph $\JM = (x, u(x))$ in the Riemannian product
space $(M \times \mathbb{R}, g + dt^2)$.
\begin{figure}[!hbt]
\centering 
\resizebox{!}{2in}{\input{pic/jang-riem}} 
\caption{}
\end{figure}
Let 
\begin{align*} 
\mc_{\hat M} &=  \Jg^{ab} \frac{D_a D_b u}{(1+D^c u D_c u)^{1/2}} \\
\tr_{\hat M} (k) &= \Jg^{ab} k_{ab} .
\end{align*} 
Then $\mc_{\hat M}$ is the mean curvature of the graph $\hat M$, 
with respect to the downward pointing normal, 
and 
$\tr_{\hat M} (k)$ is the
trace of the restriction to $\hat M$ of the pullback of $k_{ab}$ to the product $M \times
\mathbb{R}$ via the canonical projection $\pi:M\times\IR \to M$. Now
we can write Jang's equation in the form 
\begin{equation}\label{eq:jang-prime}
  \tag{\ref{eq:jang}$\prime$}
\mc_{\hat M} - \tr_{\hat M}(k) = 0
\end{equation} 
Assuming that the triple $(M,g,k)$ is the induced geometric
data for a hypersurface in a spacetime satisfying the dominant energy
condition, the induced scalar curvature is non-negative modulo a divergence
term (which of course can be large). Jang then, following the approach taken
by Geroch in the case of non-negative scalar curvature,
introduces a modified inverse mean curvature flow depending on a solution 
of Jang's equation, as well as an adapted Geroch mass that he shows to be 
formally monotone along his flow. If these steps outlined by Jang can be 
made rigorous, then his arguments lead to a proof of the positive energy 
theorem in this general situation.

Jang's work has not been developed further due to the fact that an
effective theory for existence and regularity of solutions of Jang's
equation~\eqref{eq:jang} was lacking until the work of Schoen and Yau, who applied Jang's
equation differently from the original intention by using it to
reduce the space-time positive mass theorem to the time symmetric
case. Further, it is not clear how to define an appropriate weak
solution of the modified IMCF introduced by Jang. 

\subsection{Jang's equation and positivity of mass} 
A complete proof of the positive mass theorem was first given by
Schoen and Yau \cite{Schoen-Yau:1979:PMTI}, for the special case of
time-symmetric initial data. They then extended their result to
general, asympotically flat initial data satisfying the dominant
energy condition by using Jang's equation to ``improve'' the
properties of the initial data in \cite{Schoen-Yau:1981:PMTII}. We
describe here several aspects of Jang's equation which
play a fundamental role in their work.

Firstly, Jang's equation is closely analogous to the equation 
\begin{equation}\label{eq:mots} 
\mc_\Sigma + \tr_\Sigma (k) = 0 
\end{equation} 
defining marginally outer trapped surfaces $\Sigma \subset M$, where as above $\mc_\Sigma$, $\tr_\Sigma(k)$ are the mean curvature of $\Sigma$
and the trace of $k$ restricted to $\Sigma$, respectively. 

Equations of minimal surface type 
may have blow-up solutions on general
domains and an important step in \cite{Schoen-Yau:1981:PMTII}  is 
the analysis of the blow-up sets
for the solutions of Jang's equation. 
At the boundaries of the blow-up sets, the graph of $u$ is
asymptotically vertical, asymptotic to cylinders over 
marginally outer (or inner)
trapped surfaces -- here the above mentioned relation of Jang's equation to
the MOTS equation comes into play.  

Secondly, the induced geometry of the graph $\hat M$ of a solution of
Jang's equation can be confomally changed to a metric with zero scalar
curvature without increasing the mass. 

The fundamental reason for this is that the analogue of the stability
operator for $\hat M$, i.e., the linearization of Jang's equation, has,
in a certain sense, non-negative spectrum. 
Equation (\ref{eq:jang}) is translation invariant in the vertical direction. One of
the consequences of this fact is that the non-negative  
lapse for the foliation of $M \times \mathbb{R}$ arising by
this translation can be viewed as 
a principal eigenfunction of the linearization 
of equation (\ref{eq:jang}), with eigenvalue zero. 

This spectral property allows one to prove the inequality 
\begin{equation} \label{eqn:stability_inequality_intro}
\int_{\hat {M}}  \R_{\JM} \phi^2 + 2 |D_{\JM} \phi|^2 \geq 0
\end{equation}
valid for any compactly supported Lipschitz function $\phi$
on $\hat M$, which in turn implies that the Yamabe invariant of $\JM$
is non-negative. This means that $(\hat M, \hat g)$ is conformal to a metric
of non-negative scalar curvature. In fact,
equation~\eqref{eqn:stability_inequality_intro} is stronger, since the
Yamabe operator has the factor $8$ instead of the factor $2$ in front
of the $|D_{\JM} \phi|^2$. The extra term
in~\eqref{eqn:stability_inequality_intro} with respect to the Yamabe
operator allows to control the change of the mass under this conformal
deformation. We refer to section~\ref{sec:posit-mass-theor} for the
actual calculation.  In performing the conformal transformation to
zero scalar curvature, the cylinders of the marginal boundary
components are conformally blown down to (singular) points.

The relationship between existence of solutions to Jang's equation,
the existence of MOTS, and concentration of matter was observed and
exploited in \cite{Schoen-Yau:1983:cond}.

This problem has been revisited more recently by Yau \cite{Yau:2001},
and by Galloway and O'Murchadha \cite{Galloway-OMurchadha:2008}.

In spite of a great deal of activity related to the positive mass
theorem and minimal surfaces in the years following the Schoen-Yau and
Witten proofs around 1980, little attention has been paid to Jang's
equation and MOTS from an analytical point of view until relatively
recently.

\subsection{Existence and regularity of MOTS} 
As mentioned above, the close analogy between Jang's equation, the
MOTS equation, and the minimal surface equation was exploited in
the work of Schoen and
Yau \cite{Schoen-Yau:1981:PMTII}. In particular, in that paper  
ideas from the regularity theory for minimal surfaces were applied 
to Jang's equation. The positivity of the analogue of the 
stability operator, as discussed above, plays a
central role here, completely analogous to the situation in minimal
surface theory. 
The analogy between minimal surfaces and MOTS was further developed in 
\cite{Andersson-Mars-Simon:2005}, 
where stability for MOTS was stated in
terms of non-negativity of 
the (real part of the)  
spectrum of the stability operator. The positivity property of
the stability operator for a strictly stable MOTS was used there to prove local
existence of apparent horizons. 
Further, in 
\cite{Andersson-Metzger:2005}, the curvature estimates for stable MOTS were
developed along the same lines as the regularity estimates for stable minimal
surfaces and for Jang's equation.

MOTS are not known to be stationary for an elliptic variational
problem on $(M,g,k)$. This means that the direct method of the
calculus of variation is not available to approach existence theory in
parallel with minimal or constant mean curvature surfaces. The results
in \cite{Schoen-Yau:1981:PMTII} lead Schoen \cite{Schoen:2004}
to suggest to prove existence of MOTS between a trapped and an
untrapped surface by forcing a blow-up of solutions of the
(regularized) Jang's equation. In order to carry out this program one
would like to construct a sequence of solutions to the Jang's
equations whose boundary values diverge in the limit. The physically
suggestive one-sided trapping assumptions proposed by Schoen are not
sufficient to accomplish this directly.

These technical difficulties were first overcome in
\cite{Andersson-Metzger:2009} using a bending procedure for the data
to convert the one-sided trapping assumption into the two-sided
boundary curvature conditions necessary to solve the relevant
Dirichlet problems, hence leading to a satisfying existence theory for
closed MOTS, and subsequently in an independent approach using the
Perron method in \cite{Plateau}. These two constructions have
established further features of MOTS related to their stability
\cite{Andersson-Metzger:2005,Andersson-Metzger:2009}, outward
injectivity \cite{Andersson-Metzger:2009} in low dimensions, and
almost-minimizing property \cite{Plateau}.  These properties confirm that
MOTS are in many ways very similar to minimal surfaces and surfaces
with prescribed mean curvature, which they generalize, even though
they do not arise variationally except in special cases.

The Perron method was used to solve the Plateau problem for MOTS in
\cite{Plateau} and also to extend the existence theory for closed MOTS
to more general prescribed mean curvature surfaces that do not arise
from a variational principle, including generalized apparent
horizons (see \cite{Bray-Khuri:2009A})
in \cite{GAH}. The combination of the almost minimizing property of
MOTS and the Schoen-Simon stability theory \cite{Schoen-Simon:1981}
introduced to this context in \cite{GAH} provide a convenient
framework for the analysis of MOTS in arbitrary dimension, allowing
techniques from geometric measure theory to enter despite the lack of
a variational principle.

\subsection{Overview of this paper} 
In section \ref{sec:preliminaries} we introduce notation and give some
technical preliminaries.  Section \ref{sec:theory} provides a
systematic and detailed overview of the analysis of Jang's equation
and the MOTS equation. As mentioned above, the solutions to Jang's
equation in general exhibit blow-up, and boundaries of the blow-up
regions are marginally outer (or inner) trapped. Section
\ref{sec:existence_of_MOTS_due_to_blow_up} explains how this fact can
be exploited for proving the existence of MOTS in regions whose
boundaries are trapped in an appropriate sense.  Stability of MOTS is
discussed in subsection \ref{sec:stability_of_MOTS}, where we also
describe a new result on stability of solutions to the Plateau problem for
the MOTS equation.  Section \ref{sec:applications} discusses in detail
some of the main applications of Jang's equation in general
relativity, including the positive mass theorem, formation of black
holes due to condensation of matter and the existence of outermost MOTS. Finally, section
\ref{sec:outlook} gives an overview of some open problems and
potential new applications of Jang's equation and generalizations
thereof.

%%% Local Variables: 
%%% mode: latex
%%% TeX-master: "master"
%%% ispell-dictionary: en_US
%%% End: 

\section{Preliminaries}
\label{sec:preliminaries}

\subsection{Initial data sets and MOTS}
\label{sec:initial_data_sets_and_MOTS} In this section we introduce
the notation, sign conventions, and terminology used in this
survey. Classical references for this material are \cite{HawEll} and
\cite{Wald}.

An {\it initial data set} is a triple $(M, g, k)$ where $M$
is a complete $3$-dimensional manifold, possibly with boundary,
together with a positive definite metric $g$ and a symmetric $(0, 2)$
tensor $k$. In the context of general relativity, such triples arise
as embedded spacelike hypersurfaces of time-orientable Lorentzian
manifolds $(\bar{M}, \bar g)$, referred to as the spacetime, with
induced metric $g$ and (future directed) second fundamental form
$k$. Hence, if $\eta$ is a future directed normal vector field of $M
\subset \bar{M}$ such that $\bar g(\eta, \eta) \equiv -1$, and if
$\xi, \zeta \in T_p M \subset T_p \bar M$ where $p \in M$, then
$k(\xi, \zeta) = \bar{g} (\bar D_\xi \eta, \zeta)$. Here, $\bar D$ is
the Levi-Civita connection of the spacetime.

Now let $\Sigma \subset M$ be an embedded two-sided
$2$-surface and let $\nu$ be a unit normal vector field of $\Sigma
\subset M$. We write $h$ for the second fundamental form of $\Sigma$
with respect to $\nu$ so that $h(\xi, \zeta) = g(D_\xi \nu, \zeta)$
for tangent vectors $\xi, \zeta \in T_p\Sigma$ for any $p \in \Sigma$,
where $D$ is the Levi-Civita connection of $(M, g)$. We may think of
$\nu$ as a vector field along $\Sigma \subset \bar{M}$ so that $l =\nu
+ \eta$ is a future directed null vector field of $\Sigma$ when viewed
as a spacelike $2$-surface of the spacetime. Note that then $\chi
(\xi, \zeta) := (h + k)(\xi, \zeta) = \bar g (\bar D_\xi l, \zeta)$
for tangent vectors $\xi, \zeta \in T_p \Sigma$. This symmetric $(0,
2)$-tensor $\chi$ is the null second fundamental form of $\Sigma
\subset M$ with respect to $l$. Its trace $\theta_\Sigma = \tr_\Sigma
\chi$ is called the expansion of $\Sigma$ with respect to the
null-vector field $l$. Note that $\theta_\Sigma = \mc_\Sigma +
\tr_\Sigma k$ where $\mc_\Sigma$ is the mean curvature of $\Sigma
\subset M$ with respect to the unit normal $\nu$. 

In many places in this survey, $\Sigma \subset M$ has a clearly
designated {\it outward} unit normal. If we compute the null-second
fundamental form and expansion of such a surface $\Sigma$ with respect
to the corresponding future-directed \emph{outward} null-normal, we
will write $\chi^+_\Sigma$ and $\theta^+_\Sigma$ for emphasis. If
$\Sigma$ is such a $2$-surface, with $\theta_\Sigma^+ \equiv 0$ on
$\Sigma$, then we say that $\Sigma$ is a {\it marginally outer trapped
  surface} or {\it MOTS} for short.

If we use the future-directed \emph{inward} unit normal to compute the
expansion we write $\chi^-_\Sigma$ and $\theta^-_\Sigma$.  If
$\theta^-_\Sigma \equiv 0$ we say that $\Sigma$ is a {\it marginally
  inner trapped surface} or {\it MITS}. For ease of exposition, we say
that a two-sided surface $\Sigma \subset M$ is an \emph{apparent horizon} if
it is either a MOTS or a MITS, i.e., if $\mc_\Sigma + \tr_\Sigma(k)
\equiv 0$ holds for one of the two possible consistent choices to
compute the mean curvature scalar.

The following spacetime analogue of the Bonnet-Myers theorem, whose
proof uses the Raychaudhuri rather than the Riccati equation, lies at
the heart of the Penrose-Hawking singularity theorems of general
relativity (cf. \cite[Proposition 4.4.3]{HawEll}): if the \emph{null
  energy condition} holds in $\bar{M}$, i.e. if for the spacetime
Ricci tensor $\overline \Ric(\upsilon, \upsilon) \geq 0$ holds for all
points $q \in \bar{M}$ and null vectors $\upsilon \in T_q \bar{M}$,
and if the expansion $\theta_\Sigma$ of $\Sigma \subset \bar{M}$ with
respect to the null vector field $l$ along $\Sigma$ is negative at $p
\in \Sigma$, then a null geodesic emanating from $p$ in direction
$l(p)$ has conjugate points within a finite affine
distance. Physically, this means that the surface area of a shell of
light emanating from $\Sigma$ near $p$ will go to zero before it has a
chance to `escape to infinity'.

Frequently, additional assumptions are included in the
definition of an initial data set in the literature. In this survey we
will always state such extra hypotheses explicitly when needed. Two
such extra assumptions will be particularly relevant for us. First,
recall \cite[p. 219]{Wald} that a spacetime $(\bar{M}, \bar g)$ is
said to satisfy the {\it dominant energy condition} if its
stress-energy tensor $\T := \overline{\Ric} - \frac{1}{2}\bar \R \bar
g$ has the property that for every $p \in \bar{M}$ the vector dual to
the one form $ - \T(\eta, \cdot)$ with respect to $\bar {g}$ is a
future directed causal vector in $T_p \bar{M}$ for every future
directed causal vector $\eta \in T_p \bar{M}$. Here $\overline \Ric$
and $\overline \R$ respectively denote the spacetime Ricci tensor and
spacetime scalar curvature. Note that this dominant energy condition
implies the null energy condition used above
in connection with the formation of caustics along light-like
geodesics in the spacetime. If $\eta$ arises as above as the future
directed normal vector field of a spacelike hypersurface $M \subset
\bar{M}$, then one can use the Gauss and the Codazzi equations to
express the normal-normal component (abbreviated by $\mu$) and the
normal-tangential part of $\T$ (written as a one form ${J})$ along $M$
entirely in terms of the initial data $(M, g, k)$. Explicitly,
\begin{align}
  \label{eqn:mu}
  \frac{1}{2} \left( \R_M - |k|_M^2 + (\tr_M(k))^2 \right) &=: \mu
  \\
  \div_M \left( k - \tr_M(k)g \right) &=: J
  \label{eqn:J}
\end{align}
where $\R_M$ is the scalar curvature of $(M, g)$, $D$ is its
Levi-Civita connection, and where $|k|_M^2$ and $\tr_M(k)$ are the
square length and trace of $k$ with respect to $g$. Physically, $\mu$
and $ J$ are, respectively, the energy density 
and the current
density of an observer traveling with $4$-velocity $\eta$. A
consequence of the dominant energy condition for the spacetime
$(\bar{M}, \bar g)$ is that $\mu \geq | J|$ must hold on $M$. By
slight abuse of language one says that an initial data $(M, g, p)$
satisfies the {\it dominant energy condition} if $\mu \geq | J|$
holds. An important special case of the dominant energy condition is
when $M$ is {\it a maximal slice} of the spacetime, i.e. when
$\tr_M(k) \equiv 0$, so that the dominant energy condition implies
that the scalar curvature $\R_M \geq 0$ is non-negative. In the case
where $k \equiv 0$, $M \subset \bar M$ is called a {\it time symmetric
  slice} or maybe more precisely a {\it totally geodesic slice}. In
the case of time-symmetric data $(M, g, k \equiv 0)$, the dominant
energy condition is equivalent to $R_M \geq 0$, and apparent horizons
$\Sigma \subset M$ are precisely minimal surfaces.

For our discussion of the positive mass and positive energy theorem in
subsection \ref{sec:posit-mass-theor} we will also need the following
definition (stated here as in \cite{Schoen-Yau:1981:PMTII}): An
initial data set $(M, g, k)$ is said to be \emph{asymptotically flat}
if there is a compact set $K \subset M$ so that $M \setminus K$ is
diffeomorphic to a finite number of copies of $\mathbb{R}^3 \setminus
\bar B(0, 1)$ (each corresponding to an \emph{end}), and such that
under these diffeomorphisms
\begin{align*}
  &|g_{ij} - \delta_{ij}|
  + |x| |\partial_p g_{ij}|
  + |x|^2 | \partial^2_{pq} g_{ij}|
  = O(|x|^{-1})
  \text { and }
  |\R_M| + |\partial_p \R_M|
  =
  |x|^{-4}
  \intertext { as well as }
  &|k_{ij}|
  + |x| |\partial_p k_{ij}|
  + |x|^2 |\partial^2_{pq} k_{ij}|
  = O (|x|^{-2})
  \text { and }
  \left|\sum_{i=1}^3 k_{ii}\right|
  = O(|x|^{-3})
\end{align*}
as $|x| := \sqrt{(x^1)^2 + (x^2)^2 + (x^3)^2} \to \infty$ on
each end. For the rigidity part of the positive mass theorem it will
also be necessary to assume that the third and forth order derivatives
of the metric are $O(|x|^{-4})$. 
These conditions guarantee that the ADM energy
\begin{align*}
  E_\text{ADM}
  &=
  \frac{1}{16 \pi} \lim_{r \to \infty } \sum_{i, j=1}^3
  \int_{|x| = r} \left( \partial_i g_{ij} - \partial_j g_{ii}\right) \frac{x^j}{|x|} d \mathcal{H}^2
  \intertext{and the ADM linear momentum}
  P_{\text{ADM}}^l
  &=
  \frac{1}{ 8 \pi} \lim_{r \to \infty} \sum_{j=1}^3
  \int_{|x| = r} |x|^{-1}\left( x^j k^l_j - x^l k^j_{j} \right) d\mathcal{H}^2
\end{align*}
are well-defined, see \cite{ADM:1961,
Bartnik:1986,omurchadha:1986}. 
The spacetime coordinate transformations which leave the asympototic
conditions invariant are asymptotically Lorentz transformations and the ADM
4-momentum vector $P_\text{ADM}^\mu = (E_\text{ADM},
P_\text{ADM}^l)$ is Lorentz covariant under such transformations. 
In particular, the ADM mass $m_\text{ADM} = \sqrt {-
  P_\text{ADM}^\mu (P_\text{ADM})_\mu}$ is a coordinate independent
quantity. 
Note that we will be sloppy
in the sequel and refer to the positive energy theorem, i.e. the
question whether $E_\text{ADM}\geq 0$, by the term \emph{positive mass
  theorem}, which would rather be appropriate for the statement
$m_\text{ADM}\geq 0$.

\subsection{Linearization of the expansion}
\label{sec:second_variation_formulae}
Marginally outer trapped surfaces $\Sigma \subset M \subset \bar{M}$
in general initial data sets are not known to arise as critical points
(or indeed to occur as minimizers) of a standard variational problem
described in terms of the data $(M, g, k)$. (However, note that {by
definition} a MOTS $\Sigma \subset M$ is a critical point for the area
functional inside the future-directed null-cone of $\Sigma \subset
\bar{M}$.) In recent years, properties of MOTS akin to those of
minimal surfaces have been introduced and investigated. In this
subsection we discuss the linearization of the expansion
$\theta_\Sigma$ with respect to variations in $M$. This sets the
ground for a discussion of the natural notion of stability of MOTS in
subsection \ref{sec:stability_of_MOTS}.

Let $\Sigma \subset M$ be a two-sided hypersurface and let $\varphi_s$
be a smooth family of diffeomorphisms of $M$ parametrized by $s \in
(-\delta, \delta)$ so that $\varphi_0$ is the identity. Assume that
$\frac{d}{ds}|_{s =0} \varphi_s = \sigma + f \nu$ on $\Sigma$, where
$\sigma \in \Gamma (T \Sigma)$ is a tangential vector field, where $f$
is a smooth scalar function on $\Sigma$, and where $\nu$ is a smooth unit
normal vector field of $\Sigma$. The mean curvature scalar and
tangential trace of $k$ of $\Sigma_s := \varphi_s (\Sigma)$ can be
viewed as functions on $\Sigma$ via pullback by $\varphi_s$, and they
are smooth functions of $s$. Here we agree that the mean curvature
scalar of $\Sigma$ is computed as the tangential divergence
$\div_\Sigma (\nu)$ of $\nu$. It follows that
\begin{equation}
  \label{eqn:variation_mean_curvature}
  \begin{split}
    - \la D_\Sigma \mc_\Sigma, \sigma \ra + \frac{d}{ds}|_{s = 0}
    \varphi_s^* \mc_{\Sigma_s} 
    &=
    - \Delta_\Sigma f - \left(|h|_\Sigma^2 + \Ric (\nu, \nu) \right)
    f, 
    \text{ and} 
    \\
    - \la D_\Sigma \tr_\Sigma(k), \sigma \ra + \frac{d}{ds}|_{s=0}
    \varphi_s^* \tr_{\Sigma_s} (k)
    &=
    2 k (\nu, D_{\Sigma } f) + D_\nu (\tr_M (k)) f - (D_\nu k)(\nu, \nu) f.
  \end{split}
\end{equation}
Here, $\Delta_\Sigma$ and $D_\Sigma$ denote respectively the
non-positive Laplacian and the gradient operator of $\Sigma$ with
respect to the induced metric, $|h|_\Sigma$ is the length of the
second fundamental form of $\Sigma$, and $\Ric$ and $D$ are the
ambient Ricci curvature tensor and covariant derivative operator. The
literature often refers to the first identity as the Riccati or Jacobi
equation, in particular when the tangential part $\sigma$  of the variation
 vanishes identically. If $X \in \Gamma (T \Sigma)$ denotes the
tangential part of the vector field dual to $k(\nu, \cdot)$ one can
compute that
\begin{equation*}
  (D_\nu k)(\nu, \nu)
  =
  - \mc_\Sigma k(\nu, \nu) + \la h, k \ra_\Sigma + (\div_M k) (\nu) - \div_\Sigma X.
\end{equation*}
Using the Gauss equation and the expressions for the
 mass density $\mu$ in (\ref{eqn:mu}) to substitute for
$\Ric(\nu, \nu)$, and the  current density $ J$ in (\ref{eqn:J})
to re-write $\div_M k$, it follows that
\begin{equation}
  \begin{split}
    &- \la D_\Sigma \theta_\Sigma, \sigma \ra
    + \frac{d}{ds}|_{s=0} \varphi_s^* \theta_{\Sigma_s}
    \\
    &\quad=
    - \Delta_\Sigma f + 2 \la X, D_\Sigma f \ra
    \\
    &\qquad + \left( \tfrac{1}{2} \R_\Sigma - \tfrac{1}{2} |h + k|_\Sigma^2 - J (\nu)
      - \mu + \div_\Sigma X - |X|^2 + \tfrac{1}{2} \theta_\Sigma(2 \tr_M (k)
      - \theta_\Sigma) \right) f 
    \\
    &\quad =: \L_\Sigma f 
  \end{split}
\end{equation}
where $\theta_{\Sigma_s} := \mc_{\Sigma_s} + \tr_{\Sigma_s} (k)$ is
the expansion of $\Sigma_s$. See \cite[(1)]{Andersson-Mars-Simon:2005}
and \cite[(2.3)]{Galloway-Schoen:2006}, and also
\cite[(2.25)]{Schoen-Yau:1981:PMTII} for an important special case of
this formula that we will come back to in subsection \ref{sec:stability_of_MOTS}. Following \cite{Galloway-Schoen:2006} we point
out that when $\Sigma$ is a MOTS
and $f>0$, one can rearrange and express $\frac{1}{f} \L_\Sigma f$
more compactly as
\begin{equation}
  \label{eqn:Galloway-Schoen_rearrangement}
  \div_\Sigma (X - D_\Sigma \log f) - |X - D_\Sigma \log f|_\Sigma^2 +
  \frac{1}{2} \R_\Sigma - |h + k|_\Sigma^2 - {J}(\nu) - \mu.
\end{equation}
Note that if the dominant energy condition $\mu \geq | J|$ is assumed,
then the last three terms make a non-positive contribution to this
expression.

%%% Local Variables: 
%%% mode: latex
%%% TeX-master: "master"
%%% ispell-dictionary: en_US
%%% End: 

\section{Analytical aspects of Jang's equation and MOTS}
\label{sec:theory}

\subsection {Jang's equation} \label{sec:Jangs_equation} Given a
compact subset $\Omega \subset M$ with smooth boundary, $F \in
\mathcal{C}^1 (\bar \Omega)$, and a function $\phi$ defined on
$\partial \Omega$ we consider the quasi-linear elliptic partial
differential equation expressed in local coordinates $x = (x^1, x^2,
x^3)$ as
\begin{equation}
  \label{eqn:Jangs_equation}
  \begin{split}
    \left(g^{ij} - \frac{u^i u^j}{1+|Du|^2} \right) \left(\frac {D_i D_j u}{\sqrt{1 + |Du|^2} } -
      k_{ij} \right)
    &=
    F (x) \text{ on } \Omega
    \\
    u
    &=
    \phi \text{ on } \partial \Omega 
  \end{split}
\end{equation}
where $k = k_{ij} dx^i \otimes dx^j$, $g = g_{ij} dx^i
\otimes dx^j$, $g_{ik} g^{kj} = \delta_i^j$, $u^i = g^{ij} \partial_j
u$ are the components of the gradient $D u = u^i \partial_i$ of $u$,
$|D u|^2 = {g^{ij} \partial_i u \partial_j u}$ is its length squared,
and $D_i D_j u = \partial_i \partial_j u - \Gamma^k_{ij} u_k$ are the
components of the second covariant derivative (the Hessian) of $u$ so
that $\nabla du = (D_i D_j u) dx^i \otimes dx^j$, where
$\{\Gamma_{ij}^k\}$ are the Christoffel symbols of the $g$ metric in
these coordinates. The left-hand side of (\ref{eqn:Jangs_equation}) is
independent of the choice of coordinate system and is easily seen
to be of minimal surface type \cite[Chapter 14]{GT}. Indeed, the graph
of $u$ viewed as a submanifold $\graph(u)$ of $(M \times \mathbb{R}, g
+ dt^2)$ (where $t$ is the vertical coordinate) and parametrized by an
atlas of $M$ via the defining map $x \to (x, u(x))$, has, in `base
coordinates', metric tensor $g_{ij} + \partial_i u \partial_j u$ with
inverse $g^{ij} - \frac{u^i u^j}{1 + |Du|^2}$, downward pointing unit
normal vector ${\nu} = \frac{u^i \partial_{i} - \partial_t}{\sqrt{1 +
    |Du|^2}}$, and second fundamental form $h_{ij} = \frac{D_i D_j
  u}{\sqrt{1 + |Du|^2}}$, so that the term
\begin{equation*}
  \mc(u) := \left(g^{ij} - \frac{u^i u^j }{1+|Du|^2}
  \right) \frac {D_i D_j u}{\sqrt{1 + |Du|^2} }
\end{equation*}
in (\ref{eqn:Jangs_equation}) can be interpreted
geometrically as the mean curvature of $\graph(u)$, while the
remainder
\begin{equation*}
  \tr(k)(u):=\left(g^{ij} - \frac{u^i u^j}{1+|Du|^2}
  \right) k_{ij}
\end{equation*}
computes the trace of $k$ (viewed as a tensor on the
product manifold trivially extended, so that $k (\partial_t, \cdot)
\equiv 0$) over the tangent space of $\graph(u)$. We extend these
geometric definitions in the obvious way to two-sided hypersurfaces
$\Sigma \subset M \times \mathbb{R}$ and write $\mc_\Sigma$ and
$\tr_\Sigma(k)$ (they are functions on $\Sigma$). Note that in order
to interpret the mean curvature scalar unambiguously we need to
specify a normal vector field of $\Sigma$, while $\tr_\Sigma(k)$ makes
sense independently. If $\Sigma = \graph (u)$ arises as above,
$\mc_\Sigma$ will always be computed with respect to the downward unit
normal (i.e. as its tangential divergence) so that $\mc_\Sigma (x,
u(x)) = \mc(u)(x)$. When $\Sigma$ is a hypersurface in the base $M$,
then clearly $\tr_\Sigma (k) = \tr_{\Sigma \times \mathbb{R}} (k)$ and
also $\mc_\Sigma = \mc_{\Sigma \times \mathbb{R}}$ provided the
orientations match. It is clear that, with appropriate
identifications, the geometric operators $\mc$ and $\tr(k)$ are
continuous with respect to $\mathcal{C}^2$ and $\mathcal{C}^1$
convergence of hypersurfaces respectively.

\subsection {Schoen and Yau's regularization of Jang's equation}
\label{sec:regularized_equation}

In \cite{Schoen-Yau:1981:PMTII}, R. Schoen and S.-T. Yau introduced
the geometric perspective on solutions of Jang's equation discussed in
subsection \ref{sec:Jangs_equation} and showed that solutions should
only be expected to exist in a certain sense that we will explain in
detail in subsection \ref{sec:regularization_limit}. They used their
existence theory for Jang's equation to reduce the spacetime version
of the positive energy theory (with general $k$ satisfying the
dominant energy condition) to the special case
\cite{Schoen-Yau:1979:PMTI} of the positive energy theorem where the
scalar curvature of the initial data set is non-negative. The analytic
difficulty with Jang's equation $\mc(u) - \tr(k)(u) =0$ is the lack of
an a priori estimate for $\sup_{\Omega} |u|$ due to its zero order
term $\tr_M(k) = g^{ij}k_{ij}$. That this is not just a technical
difficulty but a fundamental aspect of Jang's equation will become
apparent below in Theorems \ref{thm:eardley},
\ref{thm:condensation}. The approach of \cite{Schoen-Yau:1981:PMTII}
to bypass this issue is a positive capillarity regularization term: 
\begin{theorem} [Schoen and Yau]
  \label{thm:existence_for_regularized_equation} 
  Let $(M, g, k)$ be an initial data set, $\Omega \subset M$ a bounded
  subset with $\mathcal{C}^{2, \alpha}$ boundary, $\phi \in
  \mathcal{C}^{2, \alpha}(\partial \Omega)$, $0 < \tau \leq 1$, and
  assume that $\mc_{\partial \Omega} - |\tr_{\partial \Omega} (k)| >
  \tau \phi$ along $\partial\Omega$. Then there exists a unique solution
  $u_\tau \in \mathcal{C}^{2, \alpha} (\bar \Omega)$ of
  \begin{equation}
    \label{eqn:Dirichlet_regularized_Jangs_equation_Bdry}
    \begin{split}
      \mc(u_\tau) - \tr(k)(u_\tau)
      &=
      \tau u_\tau \text { on } \Omega
      \\
      u_\tau
      &=
      \phi \text { on } \partial \Omega.
    \end{split}
  \end{equation}
  If $(M, g, k)$ is asymptotically flat (cf. subsection
  \ref{sec:initial_data_sets_and_MOTS}), then there exists a unique
  solution $u_\tau \in \mathcal{C}^{2, \alpha}(M)$ of
  \begin{equation}
    \label{eqn:Dirichlet_regularized_Jangs_equation_AF}
    \begin{split}
      \mc(u_\tau) - \tr(k)(u_\tau)
      =
      \tau u_\tau
      \text { on } M \text{ with }
      \\
      u_\tau \to 0 \text { as } |x| \to \infty \text { on each
        asymptotically flat end.}
    \end{split}
  \end{equation}
\end{theorem}
Note that the estimates
\begin{equation*}
  \sup_\Omega \tau |u_\tau| \leq \max
  \{ 3 |k|_{\mathcal{C}(\bar \Omega)}, \tau |\phi|_{\mathcal{C}(\partial
    \Omega)} \}
\end{equation*}
respectively
\begin{equation*}
  \sup_M \tau |u_\tau| \leq 3
  |k|_{\mathcal{C}(M)}
\end{equation*}
are immediate from the maximum principle for
solutions $u_\tau$ of these regularized equations
(\ref{eqn:Dirichlet_regularized_Jangs_equation_Bdry}),
(\ref{eqn:Dirichlet_regularized_Jangs_equation_AF}). The second part
of this theorem is proven in detail in
\cite{Schoen-Yau:1981:PMTII}. The boundary gradient estimates
necessary to establish existence in the first part can be derived from
the boundary curvature condition $\mc_{\partial \Omega} -
|\tr_{\partial \Omega} (k)| > \tau \phi$ from a classical barrier
construction due to J. Serrin (as described in \cite[\S 14]{GT}): a
sufficiently ($\mathcal{C}^2$-) small monotone inward perturbation of
the boundary cylinder $\partial \Omega \times \mathbb{R}$ below the
rim $\{ (x, \phi(x)) : x \in \partial \Omega\}$ will be the graph of a
function $\underline {u}_\tau$ defined near $\partial \Omega$ which
satisfies $\mc (\underline{u}_\tau) - \tr(k) (\underline{u}_\tau) >
\tau \underline{u}_\tau$ (because $\tau>0$) with $\underline{u}_\tau =
\phi$ on $\partial \Omega$ and hence is a sub solution.  The condition
that $\mc_{\partial \Omega} + \tr_{\partial \Omega}(k) > \tau \phi$
along $\partial \Omega$ can be used to construct a super solution
$\overline{u}_\tau$ near $\partial \Omega$ by perturbing inward above
the boundary rim. For concise references and details see \cite[\S
2]{Plateau}.

\subsection{Classical results on minimal graphs and their limits}
\label{sec:minimal_graphs}

In order to motivate and explain the analysis
\cite{Schoen-Yau:1981:PMTII} of the `limits' of solutions as $\tau
\searrow 0$ of the regularized equations $u_\tau$ in Theorem
\ref{thm:existence_for_regularized_equation}, we are going to review
some classical results from minimal surface theory. First, recall that
if $u \in \mathcal{C}^2(\Omega)$ satisfies the minimal surface
equation, then $\Sigma:= \graph(u) \subset \Omega \times \mathbb{R}$
is an \emph{area minimizing boundary} in $\Omega \times
\mathbb{R}$. This means that if we denote by $E = \{ (x, t) : x \in
\Omega \text{ and } t \geq u(x)\}$ the super graph of $u$ in $\Omega
\times \mathbb{R}$, then
\begin{equation*}
  \mathcal{P} (E, W ) \leq \mathcal{P} (F, W) \text{
    for every } F \subset \Omega \times \mathbb{R} \text{ with } E \Delta
  F \subset \subset W \subset \subset \Omega \times \mathbb{R}.
\end{equation*}
Here, $\mathcal{P}(F, W)$ denotes the perimeter of the
set $F$ in $W$ (cf. \cite[p. 5]{Giusti}). This follows from a
classical calibration argument using the closed $3$-form $\eta := (d
\vol_{g + dt^2})\lfloor \nu$ where $ \nu = \frac{1}{\sqrt{1 + |Du|^2}}
(u^i \partial_i - \partial_t)$ denotes the downward normal to $\Sigma$
thought of as a unit normal vector field on $\Omega \times
\mathbb{R}$. Being an area minimizing hypersurface, $\Sigma$ is a
stable critical point of the area functional, and hence
\begin{equation}
  \label{eqn:stable_minimal_graphs}
  \int_\Sigma \left(|h|_\Sigma^2 + \Ric ( \nu, \nu) \right) \phi^2
  \leq
  \int_{\Sigma} |D_\Sigma \phi|^2 \text{ for all } \phi \in
  \mathcal{C}^1_c (\Sigma).
\end{equation}
Here $D_\Sigma \phi$ denotes the tangential gradient of
$\phi$ along $\Sigma$, $|h|_\Sigma$ is the length of the second
fundamental form, and $\Ric$ is the Ricci tensor of the ambient $M
\times \mathbb{R}$. Using the Rayleigh quotient characterization of
the first eigenvalue of an elliptic operator, one sees that the
\emph{stability inequality} (\ref{eqn:stable_minimal_graphs}) is
equivalent to the non-negativity of the Dirichlet spectrum of
$\L_\Sigma$ on compact domains of $\Sigma$, where $\L_\Sigma f := -
\left( \Delta_\Sigma + \Ric ( \nu, \nu) + |h|^2_\Sigma \right) f$ is
the linearization of the mean curvature operator, cf. subsection
\ref{sec:second_variation_formulae} (with $k \equiv 0$). We also
remind the reader that stability is implied by the existence of a
positive function $f > 0$ on $\Sigma$ such that $\L_\Sigma f \geq
0$. This can be seen either analytically, by use of the maximum
principle, or more directly by integrating the pointwise inequality
\begin{equation*}
  0
  \leq \frac{\L_\Sigma f}{f} \phi^2 = \left (\Delta_\Sigma \log f +
    |D_\Sigma \log f|^2 + |h|_\Sigma^2 + \Ric(\nu, \nu)\right) \phi^2
\end{equation*}
over $\Sigma$. Then
\begin{equation}
  \label{eq:integration_by_parts}
  \begin{split}
    \int_\Sigma \left( |h|_\Sigma^2 + \Ric (\nu, \nu)\right) \phi^2
    &=
    \int_\Sigma 2 \la  D_\Sigma \phi, \phi D_\Sigma \log f \ra - \phi^2 |D_\Sigma \log f|^2
    \\
    &
    \leq
    \int_\Sigma |D_\Sigma \phi|^2
  \end{split}
\end{equation}
follows from an integration by parts and an application
of Young's inequality. The vector field $- \partial_t$ generates a
family of diffeomorphisms $\varphi_s$ that act by downward
translation. Clearly, the variations of $\Sigma$ by $\varphi_s$ also
have zero mean curvature, and hence it follows from subsection
\ref{sec:second_variation_formulae} that $\L_\Sigma \left( \frac{1}{v}
\right) = 0$, where as usual $v = \sqrt{1 + |Du|^2}$. Hence
$\frac{1}{v}$ is a positive Jacobi field for $\L_\Sigma$ on
$\Sigma$. Our point here is that the stability property of minimal graphs as
expressed by (\ref{eqn:stable_minimal_graphs}) can be recovered
without recourse to their stronger minimizing property. We emphasize
that this Jacobi field $\frac{1}{v}$ measures the vertical component
of the (upward) unit normal of $\Sigma$. 

Minimizing boundaries are smooth in ambient dimension $\leq
7$ and they form a closed subclass $\mathcal{F}$ of boundaries of
locally finite perimeter in $\Omega \times \mathbb{R}$ with respect to
current convergence (see e.g. \cite{Giusti}). It follows that every
sequence of minimal graphs has a smooth subsequential limit $\Sigma_ i
\to \Sigma$ in $\mathcal{F}$. If $\frac{1}{v_i}$ denotes as above the
vertical part of the (upward) unit normal of $\Sigma_i$, then
$\L_{\Sigma_i} v_i^{-1}= 0$, and hence $\Delta_{\Sigma_{i}}
\frac{1}{v_i} \leq \frac{\beta}{v_i}$ holds where $\beta =
|\Ric|_{\mathcal{C}(\Omega)}$ is a constant independent of $i$. This
differential inequality has a non-parametric interpretation: the
vertical part of the unit normal of the hypersurface $\Sigma_i$ is a
non-negative super solution of a geometric homogeneous elliptic
equation on $\Sigma_i$. This aspect of minimal graphs passes to their
subsequential limit $\Sigma$. The Hopf maximum principle then implies
that on every connected component of $\Sigma$ the vertical part of the
unit normal either has a sign or vanishes identically. Put
differently: subsequential limits of minimal graphs are minimizing
boundaries in $\Omega \times \mathbb{R}$ whose components are either
graphical over open subsets of $\Omega$ or vertical cylinders whose
cross-sections are minimizing boundaries in the base $\Omega$. This
analysis of the limiting behavior of minimal graphs is carried out in
detail in \cite{MM}, where concise references can be found. 

\subsection{\texorpdfstring{Behavior of $\graph(u_\tau)$ in the regularization limit}{Behavior in the regularization limit}}
\label{sec:regularization_limit}

We consider a sequence of solutions $\{u_\tau\}_{0 < \tau \leq 1}
\subset \mathcal{C}^{2, \alpha}(\Omega)$ of $\mc(u_\tau) -
\tr(k)(u_\tau) = \tau u_\tau$ as in Theorem
\ref{thm:existence_for_regularized_equation}.

For general $k$, their graphs $\Sigma_\tau$ do not satisfy an
apparent, useful variational principle. However, using the Jacobi
equation to compute the variation of the mean curvature of
$\Sigma_\tau$ with respect to vertical translation, one obtains that
\begin{equation*}
  \left( \Delta_{\Sigma_\tau} + \Ric ( \nu_\tau,\nu_\tau) +
    |h_\tau|_{\Sigma_\tau}^2 \right) \frac{1}{v_\tau} 
  =
  -\nu_\tau \left( \mc (u_\tau) \right) \frac{1}{v_\tau}
\end{equation*}
where as before $v_\tau:= \sqrt{1 + |Du_\tau|^2}$, $h_\tau$
denotes the second fundamental form of $\Sigma_\tau$ and
$|h_\tau|_{\Sigma_\tau}$ is its length, $\nu_\tau = \frac{1}{v_\tau}
(u_\tau^i \partial_i - \partial_t)$ is the downward unit normal of
$\Sigma_\tau$, $\Ric$ is the Ricci tensor of $(M \times \mathbb{R}, g
+ dt^2)$, and where we differentiate $\mc(u_\tau)$ on the right
as a function on the base $\Omega$. Expanding the right-hand side
using the equation for $u_\tau$, one obtains
\begin{equation*}
  \left( \Delta _{\Sigma_\tau} + \Ric(\nu_\tau,
    \nu_\tau) + |h_\tau|_{\Sigma_\tau}^2 \right) \frac{1}{v_\tau} \leq 10
  \left( |k|_M |h_\tau|_{\Sigma_\tau} + |D k|_M \right)
  \frac{1}{v_\tau}.
\end{equation*}
Multiply this pointwise inequality
by $v_\tau \phi^2$ where $\phi \in \mathcal{C}^1_c(\Sigma_\tau)$,
integrate over $\Sigma_\tau$, and integrate by parts as in
(\ref{eq:integration_by_parts}) to see that there exists a constant
$\beta = \beta (|k|_{\mathcal{C}^1(\Omega)})$ so
that
\begin{equation*}
  \begin{split}
    \int_{\Sigma_\tau} \left( \Ric (\nu_\tau, \nu_\tau)
      + |h_\tau|_{\Sigma_\tau}^2 \right) \phi^2
    &\leq
    \int_{\Sigma_\tau} |D_{\Sigma_\tau} \phi|^2 + \beta \int_{\Sigma_\tau}
    (|h_\tau|_{\Sigma_\tau} + 1) \phi^2
    \\
    &\qquad\qquad\text{ for all } \phi \in
    \mathcal{C}^1_c(\Sigma_\tau).
  \end{split}
\end{equation*}
This almost looks like the stability inequality
(\ref{eqn:stable_minimal_graphs}) for minimal graphs. Schoen and Yau
observed that this inequality, together with the a priori estimate for
$\tau |u_\tau|$ that follows from the maximum principle in the
situation of Theorem \ref{thm:existence_for_regularized_equation}, the
bound for $\mc(u_\tau)$ hence resulting from the defining equation for
$u_\tau$, and a local area bound for $\Sigma_\tau$ that follows from
this bound on the mean curvature and a calibration argument as in
subsection \ref{sec:minimal_graphs}, can be used to derive pointwise
curvature estimates for $\Sigma_\tau$ as in
\cite{Schoen-Simon-Yau}. Moreover, they derived a geometric Harnack
inequality for the vertical part $\frac{1}{v_\tau}$ of the unit normal
of $\Sigma_\tau$ in the effective form $\sup_{B_\rho \cap \Sigma_\tau}
\frac{1}{v_\tau} \leq \gamma \inf_{B_\rho \cap \Sigma_\tau}
\frac{1}{v_\tau}$ for any extrinsic ball $B_\rho$ centered on
$\Sigma_\tau$ and such that $B_{2 \rho} \subset \Omega \times
\mathbb{R}$. 

These results of \cite{Schoen-Yau:1981:PMTII} replace the
variational arguments for minimal graphs in subsection
\ref{sec:minimal_graphs}. Hence, given a sequence $\{ u_\tau \}_{0 <
  \tau \leq 1}$ of solutions of
(\ref{eqn:Dirichlet_regularized_Jangs_equation_Bdry}) so that
$\sup_{\Omega, \tau} \tau |u_\tau| < \infty$, their graphs
$\Sigma_\tau$ have a smooth embedded subsequential limit $\Sigma$ as
$\tau \searrow 0$.  (The curvature estimates depend on the distance to
the boundary of $\Omega$. Boundary barriers as in the proof of Theorem
\ref{thm:existence_for_regularized_equation} can be used to show that
$\Sigma_\tau$ must remain bounded near $\partial \Omega \times
\mathbb{R}$ so that standard PDE techniques can be used to analyze the
limit of $u_\tau $ near $\partial \Omega$.) From the Harnack estimate
above it is then not difficult to see that the components of the limit
$\Sigma$ are complete, embedded, and separated from one another by a
positive distance, and that each is either graphical or
cylindrical. 

Obviously, $\mc_\Sigma - \tr_\Sigma(k) = 0$ must hold and
it follows that the graphical components in the limit are solutions of
Jang's equation, and that the cross-sections of the cylindrical
components are apparent horizons of the base. These cross-sections
inherit an orientation from the surfaces $\Sigma_\tau$ of the
regularization limit, depending on whether these graphs `blow-up' or
`blow-down' along them. The union of the graphical components of
$\Sigma$ is the graph of a function $u$ defined on an open set
$\Omega_0 \subset \Omega$, and this $u$ must be unbounded on approach
to $\partial \Omega_0 \setminus \partial \Omega$. Moreover, $\Sigma_0$
smoothly asymptotes the boundary cylinder above the respective
component of $\partial \Omega_0 \setminus \partial \Omega$, diverging
either to positive or negative infinity, and it follows that $\partial
\Omega_0 \setminus \partial \Omega$ is made up of apparent
horizons.

In summary, one has the following general existence result:
\begin{theorem}[\cite{Schoen-Yau:1981:PMTII}]
  \label{thm:existence_for_Jangs_equation}
  Let $(M, g, k)$ be a complete $3$-dimensional initial data set, let
  $\Omega \subset M$ be a bounded open subset with $\mathcal{C}^{2,
    \alpha}$ boundary such that $\mc_{\partial \Omega} > |\tr_{\partial
    \Omega}(k)|$ holds along $\partial \Omega$ where $\mc_{\partial
    \Omega}$ is the mean curvature of $\partial \Omega$ computed as the
  tangential divergence of the unit normal pointing out of $\Omega$, and
  where $\tr_{\partial \Omega} (k)$ is the trace of $k$ over the tangent
  space of $\partial \Omega$, and let $\phi \in \mathcal{C}^{2, \alpha}
  (\partial \Omega)$.

  Then there exists an open subset $\Omega_0 \subset
  \Omega$ with embedded boundary consisting of $\partial\Omega$ and of
 the union of two finite, possibly empty collections of smooth disjoint connected
  closed apparent horizons $\{\Sigma^+_i\}$ and $\{\Sigma_j^-\}$, as
  well as a function $u \in \mathcal{C}^{2, \alpha} (\partial \Omega
  \cup \Omega_0)$ so that
  \begin{align*}
    \mc (u) - \tr(k)(u)
    &= 0
    \text { on } \Omega_0
    \\
    u &= \phi \text{ on } \partial \Omega
    \\
    u(x) &\to + \infty
    \text { uniformly as } \dist(x, \Sigma^+_i) \to 0,
    \text { and }
    \\
    u(x) &\to - \infty
    \text { uniformly as } \dist (x, \Sigma^-_j) \to 0.
  \end{align*}
  Computing the mean curvature
  scalar with respect to the unit normal that points into $\Omega_0$,
  the surfaces $\Sigma^+_i$ satisfy $\mc_{\Sigma^+_i} +
  \tr_{\Sigma^+_i}(k) = 0$ while $\mc_{\Sigma^-_j} - \tr_{\Sigma^-_j}(k)
  = 0$. The same conclusion holds if $(M, g, k)$ is asymptotically flat,
  $\Omega = M$, and where $u(x) \to 0$ on each of the asymptotically
  flat ends.
\end{theorem}

\subsection {Existence of MOTS due to blow-up}
\label{sec:existence_of_MOTS_due_to_blow_up}

The analysis of \cite[Proposition 4]{Schoen-Yau:1981:PMTII} outlined
in the previous subsection contains the following, general result:
given a sequence of functions $\{u_\tau\}_{0 < \tau \leq 1} \subset
\mathcal{C}^{2, \alpha}(\Omega)$ with $\mc(u_\tau) - \tr(k)(u_\tau) =
\tau u_\tau$ and $\sup_{\Omega, \tau} \tau |u_\tau| < \infty$, then
for every $\Omega' \subset \subset \Omega$ there exists a sequence
$\tau_i \searrow 0$ so that the graphs $\Sigma_{\tau_i}$ converge to a
smooth hypersurface $\Sigma \subset \Omega' \times \mathbb{R}$. This
limit $\Sigma$ is an apparent horizon in the initial data set $(M
\times \mathbb{R}, g + dt^2, k)$, and each component of $\Sigma$ is
either graphical or is a vertical cylinder whose cross-section is an
apparent horizon in the base $\Omega'$. The sets $\Omega_{\pm} := \{x
\in \Omega: \lim \sup_{i \to \infty} \pm u_{\tau_i} (x) = \pm
\infty\}$ are disjoint and relatively open in $\Omega'$, and their relative boundaries in
$\Omega'$ are smooth, embedded apparent horizons. Let
$\Omega_0 := \Omega \setminus (\bar \Omega_+ \cup \bar \Omega _-)$. The
union of the graphical components of $\Sigma$ is a graph $u : \Omega'
\cap \Omega_0 \to \mathbb{R}$ solving Jang's equation $\mc(u) -
\tr(k)(u) = 0$. The union of the cylindrical components of $\Sigma$ is
given by $\Omega' \cap (\partial \Omega_+ \cap \partial \Omega_-)
\times \mathbb{R}$. The function $u$ tends to $\pm \infty$ near
$\Omega' \cap (\partial \Omega_0 \cap \partial \Omega_\pm)$ and its
graph is smoothly asymptotic as a submanifold to the vertical cylinder
based on this set. In particular, if the sequence $u_\tau(x)$ diverges
to $+ \infty$ for some $x \in \Omega$ as $\tau \searrow 0$ while
staying finite or diverging to $- \infty$ at other points of $\Omega$,
then there must be (part of) an apparent horizon in $\Omega$.

Using this blow-up analysis of Schoen and Yau and Theorem
\ref{thm:existence_for_regularized_equation} for the existence of
solutions for the regularized Jang's equation, it is not difficult to
see that (non-empty!) closed apparent horizons $\Sigma \subset \Omega$
exist in every bounded subset $\Omega \subset M$, provided that
$\partial \Omega$ has at least two boundary components and satisfies
$\mc _{\partial \Omega} - |\tr_{\partial \Omega}(k)| > 0$ (say greater
than some small $\varepsilon > 0$). Simply consider solutions $u_\tau
: \bar \Omega \to \mathbb{R}$ of $\mc(u_\tau) - \tr(k)(u_\tau) = \tau
u_\tau$ with $u_\tau = \frac{\varepsilon}{\tau}$ on one of the
boundary components and $u_\tau = - \frac{\varepsilon}{\tau}$ on the
others. From the maximum principle one has that $\sup_{\bar \Omega,
  \tau} \tau |u_\tau| \leq \max\{ \varepsilon, 3 |k|_{\mathcal{C}(\bar
  \Omega)}\}$. The barriers used in the proof of Theorem
\ref{thm:existence_for_regularized_equation} can be used to show that
these solutions $u_\tau$ diverge in a fixed neighborhood of the
boundary.

Schoen \cite{Schoen:2004} 
suggested that the existence of MOTS should even
follow if one only assumes that some boundary components of $\Omega$,
whose union we denote by $\partial_+ \Omega$, satisfy $\mc_{\partial_+
  \Omega} - \tr_{\partial_+ \Omega} (k) > 0$, while $\mc_{\partial_-
  \Omega} + \tr_{\partial_- \Omega} (k) > 0$ holds for the union
$\partial_- \Omega$ of all other boundary components, because these
conditions give rise to upper barriers for `blow-up' near $\partial_+
\Omega$ and to lower barriers for `blow-down' near $\partial_-
\Omega$. There is an important technical difficulty in implementing
this approach: these one-sided barriers do not guarantee that
solutions of the boundary value problems
(\ref{eqn:Dirichlet_regularized_Jangs_equation_Bdry}) exist
classically. 

That these \emph{technical} difficulties can be overcome has
been shown first by \cite{Andersson-Metzger:2009} and then by
\cite{Plateau} using two independent methods which exhibit different
features for the MOTS that are shown to exist.

The `bending of the boundary data' approach of
\cite{Andersson-Metzger:2009} shows that the metric $g$ and the second
fundamental form tensor $k$ of the initial data set $(M, g, k)$ can be
modified to $\tilde g$ and $\tilde k$ in a neighborhood of $\partial_+
\Omega$ so that the region where the changes takes place is foliated
by surfaces $\Sigma_s$ with positive expansion $\tilde \mc_{\Sigma_s}
- \tr_{\Sigma_s}(\tilde k)>0$ but so that $\tilde k$ vanishes
identically near $\partial_+ \Omega$. A similar modification can be
made near $\partial_-\Omega$. The above argument can then be used to
show that a MOTS $\Sigma$ exists in the modified data set. The maximum
principle shows that this $\Sigma$ cannot intersect the region where
the data has been modified so that $\Sigma$ is also a MOTS with
respect to the original data set. The change of the geometry and the
data are local but large. In particular, the change of the quantity
$|\tilde k|_{\mathcal{C}(\bar \Omega)}$ that many of the geometric
estimates for $\Sigma$ in subsection \ref{sec:regularization_limit}
depend on is hard to control explicitly. Importantly though,
\cite[\S4]{Andersson-Metzger:2009} develop a delicate barrier argument
that shows that the MOTS $\Sigma$ arising in the blow-up is stable in
the sense of MOTS as discussed in subsection
\ref{sec:stability_of_MOTS}. The curvature estimates of
\cite{Andersson-Metzger:2005}, which do not require a priori area
bounds and which depend only on the original data, are then available
for $\Sigma$.

An alternative line of proof is given in \cite{Plateau},
where the Perron method was introduced to the analysis of Jang's
equation (cf. \cite{Serrin} for its classical application to minimal
and constant mean curvature graphs) and used to generate maximal
interior solutions $u_\tau$ for the boundary value problems
(\ref{eqn:Dirichlet_regularized_Jangs_equation_Bdry}) with estimates
$\sup_{\bar \Omega, \tau} \tau |u_\tau| \leq \max\{ \varepsilon, 3
|k|_{\mathcal{C}(\bar \Omega)}\}$. These solutions won't assume
particular boundary values in general, but they will lie above
(respectively below) the lower (upper) barrier constructed from the
boundary curvature conditions, which is all that is needed in the
argument to force divergence near the boundary. Instead of using
stability-based curvature estimates, a variant of the calibration
argument mentioned in subsection \ref{sec:minimal_graphs} is applied to show that the surface $\Sigma$ constructed in
the process is a $C$-almost minimizing boundary in $\Omega$. By this
we mean that $\Sigma$ is the boundary of a set $E$ in $\Omega$ such
that
\begin{equation}
  \begin{split}
    \mathcal{P} (E, W ) \leq \mathcal{P} (F, W) + C |E \Delta F|
    \\
    \text{ for every } F \subset \Omega \text{ such that } E \Delta F \subset \subset W \subset \subset \Omega
  \end{split}
\end{equation}
where $C := 6 |k|_{\mathcal{C}(\bar \Omega)}$. (See \cite{DuzaarSteffen} for a
systematic study of such almost minimizing boundaries, and
\cite[Appendix A]{Plateau} for further concise references on the
relevant geometric measure theory.) So
$\Sigma$ minimizes area in $\Omega$ modulo a lower order bulk term
that is controlled explicitly. This feature of $\Sigma \subset \Omega$
is inherited from an analogous property of $\graph(u_\tau)$, see
\cite[Example A.1]{Plateau} for details. The stability-based curvature
estimates in \cite{Schoen-Yau:1981:PMTII} are replaced by techniques
from geometric measure theory which allow transitioning of the
argument to high dimensional initial data sets if one accepts a thin
singular set, as with minimal surfaces. 

In conclusion, we have the following existence theorem combining
\cite{Andersson-Metzger:2009} and \cite{Plateau}:
\begin{theorem}
  \label{thm:Schoens_theorem}
  Let $(M, g, k)$ be a 3-dimensional
  initial data set and let $\Omega \subset M$ be a connected bounded
  open subset with smooth embedded boundary $\partial \Omega$. Assume
  this boundary consists of two non-empty closed hypersurfaces
  $\partial_+ \Omega$ and $\partial_- \Omega$, possibly consisting of
  several components, so that
  \begin{equation}
    \label{eq:trapping_assumption_for_Schoens_theorem}
    \mc_{\partial_+ \Omega} - \tr_{\partial_+ \Omega} k > 0 \text{ and }
    \mc_{\partial_- \Omega} + \tr_{\partial_- \Omega} k > 0,
  \end{equation}
  where the mean curvature scalar is computed as the
  tangential divergence of the unit normal vector field that is pointing
  out of $\Omega$. Then there exists a smooth closed embedded
  hypersurface $\Sigma \subset \Omega$ homologous to $\partial_- \Omega$
  such that $\mc_\Sigma + \tr_\Sigma (k) = 0$ (where $\mc_\Sigma$ is
  computed with respect to the unit normal pointing towards
  $\partial_-\Omega$). $\Sigma$ is stable in the sense of MOTS and it is
  $C$-almost minimizing in $\Omega$ for an explicit constant $C =
  C(|k|_{\mathcal{C}(\bar \Omega)})$. This existence result and all the
  properties listed above carry over to initial data sets of dimensions
  $\leq 7$. In dimensions greater than $7$ we have the existence of a
  $C$-almost minimizing boundaries $\Sigma$ in $\Omega$ with a singular
  set of Hausdorff codimension at most $7$ that satisfy the marginally
  outer trapped surface equation distributionally.
\end{theorem}

The approach via the Perron method can easily be adapted to prove
existence of surfaces $\Sigma$ whose mean curvature is prescribed as a
continuous function of position and unit normal under boundary
curvature conditions analogous to those in Theorem
\ref{thm:Schoens_theorem}, see \cite{GAH}, 
recovering classical existence results for variational prescribed mean
curvature problems in special cases.  
In \cite{Plateau}, the Perron method has
been used in conjunction with the analysis of Schoen and Yau described
above to prove existence of MOTS spanning a given boundary curve, in
analogy with the classical Plateau problem for minimal surfaces. We
describe the general result for $n$-dimensional initial data sets:

\begin{theorem}
  \label{thm:Plateau_problem}
  Let $(M^n, g, k)$ be an
  initial data set and let $\Omega \subset M^n$ be a bounded open domain
  with smooth boundary $\partial \Omega$. Let $\Gamma^{n-2}
  \subset \partial \Omega$ be a non-empty, smooth, closed, embedded
  submanifold that separates this boundary in the sense that $\partial
  \Omega \setminus \Gamma^{n-2} = \partial_- \Omega \dot \cup \partial_+
  \Omega$ for disjoint, non-empty, and relatively open subsets
  $\partial_- \Omega, \partial_+ \Omega$ of $\partial \Omega$. Assume
  that $\mc_{\partial \Omega} - \tr_{\partial \Omega} k > 0$ near
  $\partial_+ \Omega$ and that $\mc_{\partial \Omega} + \tr_{\partial
    \Omega} k > 0$ near $\partial_- \Omega$ where the mean curvature
  scalar is computed as the tangential divergence of the unit normal
  pointing out of $\Omega$. Then there exists an almost minimizing
  (relative) boundary $\Sigma^{n-1} \subset \Omega$, homologous to
  $\partial_- \Omega$, with singular set strictly contained in $\Omega$
  and of Hausdorff dimension $\leq (n-8)$, so that $\Sigma^{n-1}$
  satisfies the equation $\mc_\Sigma + \tr_\Sigma(k) = 0$
  distributionally, and so that $\Sigma^{n-1}$ is a smooth hypersurface
  near $\Gamma^{n-2}$ with boundary $\Gamma^{n-2}$. In particular, if $2
  \leq n \leq 7$, then $\Sigma^{n-1}$ is a smooth embedded marginally
  outer trapped surface in $\Omega$ which spans
  $\Gamma^{n-2}$.
\end{theorem}

We conclude this section by noting the close relation of the features
of the regularization limit of Jang's equation from
\cite{Schoen-Yau:1981:PMTII} with the classical Jenkins-Serrin theory
\cite{JenkinsSerrin68} of finding Scherk-type (i.e. infinite boundary
value) minimal graphs in polygonal regions in $\mathbb{R}^2$ and its
obstructions, as expressed by the Jenkins-Serrin conditions, and also
the extensions of this theory to infinite boundary value constant mean
curvature graphs \cite{Serrin}, \cite{Spruck72} and \cite{HRS08} (see
also references therein) in curvilinear domains in $\mathbb{R}^2$,
$\mathbb{S}^2$ and $\mathbb{H}^2$.

\subsection{Stability of MOTS and an identity of Schoen and Yau}
\label{sec:stability_of_MOTS}

\begin{definition}[\cite{Andersson-Mars-Simon:2005},\cite{Andersson-Metzger:2005}]
  \label{def:stability_of_MOTS}
  A closed
  two-sided surface $\Sigma \subset M \subset \bar{M}$ with vanishing
  expansion $\theta_\Sigma \equiv 0$ (computed with respect to the
  future-directed null normal $l=\eta + \nu \in \Gamma (\Sigma, T
  \bar{M})$ where $\nu$ is a designated `outward' unit normal vector
  field of $\Sigma \subset M$) is said to be {\it stable in the sense of
    MOTS} if there exists a positive function $f > 0$ on $\Sigma$ so that
  $\L_\Sigma f \geq 0$, where
  \begin{equation*}
    \L_\Sigma f := - \Delta_\Sigma f + 2 \la X,
    D_\Sigma f \ra + \left( \tfrac{1}{2} \R_\Sigma - \tfrac{1}{2} |h +
      k|_\Sigma^2 - J (\nu) - \mu + \div_\Sigma X - |X|^2 \right) f
  \end{equation*}
  Here, $X$ is the tangential part of the one form dual
  to $k(\nu, \cdot)$ on $\Sigma$ and $\R_\Sigma$ is the scalar curvature
  of $\Sigma$.
\end{definition}
A few remarks are in order. First, note that $\L_\Sigma f$ here is the
linearization of the expansion $\theta^+$ for normal perturbations $f
\nu$ of $\Sigma$, cf. subsection
\ref{sec:second_variation_formulae}. When $k \equiv 0$, then this
definition is consistent with the usual strong stability condition for
closed minimal surfaces, as can be seen using the argument in
subsection \ref{sec:minimal_graphs}. Note that if we had $f>0$ with
strict inequality $\L_\Sigma f > 0$, then a neighborhood of $\Sigma$
in $M$ could be foliated by surfaces $\{\Sigma_s\}_{- \delta < s <
  \delta}$ with $\Sigma_0 = \Sigma$ and so that $\Sigma_{s}$ lies
`outside' of $\Sigma$ with respect to $\eta$ and has positive
expansion $\theta_{\Sigma_s} >0$ for $0 < s < \delta$, and such that
$\Sigma_s$ lies inside of $\Sigma$ and has negative expansion for $-
\delta < s < 0$. (For minimal surfaces this {\it strict stability
  condition} implies that the surface is minimizing in this
neighborhood.) Note that in general the operator $\L_\Sigma$ is not
self-adjoint. It was noted in \cite[Lemma
1]{Andersson-Mars-Simon:2005} that by the Krein-Rutman theorem the
eigenvalue $\lambda$ of $\L_\Sigma$ with the least real part is real,
and that there exists an eigenfunction $\L_\Sigma \Theta = \lambda
\Theta$, positive on at least one and vanishing on all other connected
components of $\Sigma$, corresponding to this \emph{principal}
eigenvalue $\lambda$. The maximum principle then implies that the
condition in Definition \ref{def:stability_of_MOTS} is equivalent to
$\lambda \geq 0$. As in subsection
\ref{sec:second_variation_formulae}, it is useful to rewrite the
pointwise condition $0 \leq \frac{1}{f} \L_\Sigma f$ of Definition
\ref{def:stability_of_MOTS} as
\begin{equation}
  \mu + {{J}}(\nu) \leq \div_\Sigma (X - D_\Sigma \log
  f) - |X - D_\Sigma \log f|_\Sigma^2 + \frac{1}{2} \R_\Sigma -
  \frac{1}{2} |h + k|_\Sigma^2.
\end{equation}
If the dominant energy condition $\mu \geq | J|$ holds, then
the left-hand side here is non-negative, and an integration by parts
exactly as in (\ref{eq:integration_by_parts}) implies that
\begin{equation}
  \label{eq:integrated_Schoen-Yau_identity} \int_\Sigma
  \frac{1}{2} |h+k|^2_\Sigma \leq \int_\Sigma \frac{1}{2} \R_\Sigma
  \phi^2 + |D_\Sigma \phi|^2 \text{ for every } \phi \in \mathcal{C}^1_c
  (\Sigma).
\end{equation}
Together with the Gauss-equation one can conclude that
\begin{equation}
  \label{eq:rough_consequence_of_integrated_Schoen_Yau_identity}
  \int_\Sigma |h|^2_\Sigma \phi^2
  \leq
  \int_\Sigma |D_\Sigma \phi|^2 + \beta \int_\Sigma (|h|_\Sigma + 1) \phi^2
\end{equation}
where $\beta$ only depends on the initial data set but not on
$\Sigma$.

The following example of stable MOTS is one of the key observations in
\cite{Schoen-Yau:1981:PMTII}. Let $u: \Omega_0 \cup \partial \Omega
\to \mathbb{R}$ be a graphical solution to Jang's equation in the
sense of subsection \ref{sec:regularization_limit}. So $\Omega_0
\subset \Omega$, the boundary of $\Omega_0$ consists of $\partial
\Omega$ together with a finite number of smooth embedded apparent
horizons, $\mc (u) - \tr(k)(u) = 0$ on $\Omega_0$, $\Sigma :=
\graph(u) \subset M \times \mathbb{R}$ is a complete submanifold with
boundary $\{(x, u(x)) : x \in \partial \Omega \}$, and $u$ diverges
uniformly on approach to the components of $\partial \Omega_0
\setminus \partial \Omega$. As discussed in subsection
\ref{sec:Jangs_equation}, $\Sigma$ has vanishing expansion in the
initial data set $(M \times \mathbb{R}, g + dt^2, - k)$ (mind the
sign!), and $\L_\Sigma f \equiv 0$ where $0 < f =
\frac{1}{\sqrt{1+|Du|^2}} = \la- \partial_t, \nu\ra$ is the normal
component of the unit vector field generating downward translation. As
above one has
\begin{equation}
  \label{eqn:Schoen-Yau_identity}
  \mu -{{J}}(\nu) =
  \div_\Sigma (X - D_\Sigma \log f) - |X - D_\Sigma \log f|_\Sigma^2 +
  \frac{1}{2} \R_\Sigma - |h - k|_\Sigma^2.
\end{equation}
where $\mu$ and ${J}$ are computed with respect to $(M
\times \mathbb{R}, g + dt^2, k)$ (sorry!) and where $X$ is the
tangential part of the one form dual to $- k(\nu, \cdot)$ (we triple
checked this sign). Note that $\mu$ does not depend on the $t$-coordinate and
coincides with the mass density of $(M, g, k)$, and that the
same holds for $ {J}$. This is equation (2.25) in
\cite{Schoen-Yau:1981:PMTII} where it was derived by direct
computation. See also equation (18) in \cite{Jang:1978} where the
identity appears in disguised form and without geometric
interpretation. 

It has been known from \cite{Schoen-Yau:1981:PMTII} that closed MOTS
$\Sigma \subset M$ that arise in the regularization limit of Jang's
equation are ``symmetrized stable,'' i.e., the operator
$\L^{\text{sym}}_\Sigma f := - \Delta_\Sigma f + \left( \frac{1}{2}
  \R_\Sigma - \frac{1}{2} |h + k|_\Sigma^2 - J (\nu) - \mu \right) f$
on $\Sigma$ has non-negative spectrum. In
\cite{Andersson-Metzger:2009} it was proven such surfaces are stable
in the sense of MOTS, which is a stronger \cite[Lemma 2.2]
{Galloway:2008} and physically more conclusive result. Here we discuss
this stability from a geometric point of view, and we also discuss the
stability of MOTS solving the Plateau problem in \cite{Plateau}.

Let $\Sigma \subset M$ be a connected \emph{closed}
two-sided unstable MOTS with respect to the unit normal $\nu$. The
Krein-Rutman theorem implies (cf. \cite{Andersson-Mars-Simon:2005})
that there is $\lambda < 0$ and a strictly positive function $\Theta
\in \mathcal{C}^{\infty}(\Sigma)$ so that $\L_\Sigma \Theta = \lambda
\Theta$. The stability operator of $\Sigma \times \mathbb{R}$ with
respect to the extended initial data set $(M \times \mathbb{R}, g +
dt^2, k)$ is $\L_{\Sigma \times \mathbb{R}} = - \frac{d^2}{dt^2} +
\L_\Sigma$ (where $\L_\Sigma$ ignores the dependence on the vertical
variable $t$). Note that if $T = T(t)$ is a function $T \in
\mathcal{C}^2(\mathbb{R})$ then $\L_{\Sigma \times \mathbb{R}} (
\Theta T) = \Theta ( \lambda T - \frac{d^2}{dt^2} T)$. Consider the
function $T(t) = - \varepsilon \left( 1 - \exp (\frac{t}{N})\right)$
where $\varepsilon > 0$ is small and $N>1$ is large. The relevant
properties of $T$ are that $T(0) = 0$, that $T'<0$, that $T(t) \to -
\varepsilon$ as $t \to - \infty$, and that $- T'' + \lambda T \geq -
\frac{\varepsilon \lambda}{2} > 0$ for all $t \in (-\infty, 1]$
provided $N$ is sufficiently large (depending only on
$\lambda$). Hence $\L_{\Sigma \times \mathbb{R}} (\Theta T) \geq -
\Theta \frac{\varepsilon \lambda}{2} \geq \eta$ in this range for a
positive constant $\eta>0$. This means that for $s >0$ sufficiently
small, the surface $\{ \exp_{(\theta, t)} \left( s \Theta(\theta) T(t)
  \nu \right) : (\theta, t) \in \Sigma \times (-\infty, 1) \}$ is a
smooth hypersurface (with boundary) in $M \times \mathbb{R}$ that has
\emph{positive} expansion everywhere. Since $T$ is monotone this
hypersurface can be written as the graph of a function $\tilde{u} : U
\to \mathbb{R}$ where $U$ is an \emph{open} neighborhood of $\Sigma$
such that $\tilde{u} > 0$ in the part of $U$ that lies to the side of
$\nu$ and so that $\tilde{u} \to - \infty$ on approach to the part of
the boundary of $U$ that lies in direction $-\nu$ as seen from
$\Sigma$. For $\bar u := - \tilde u$ we have that $\mc(\bar {u}) -
\tr(k)(\bar {u}) < 0$ is a super solution of Jang's equation. Using
$-T$ instead of $T$ one obtains a sub solution $\underline u $ of
Jang's equation with analogous properties. (The awkward sign reversal
here is due to the fact that Jang's equation is the MITS equation
rather than a MOTS equation with respect to the data set $(M \times
\mathbb{R}, g+ dt^2, k)$.) 

There are three situations in which a closed MOTS $\Sigma$
can arise in the regularization limit of Jang's equation: $\Sigma
\subset \partial \Omega_0 \cap \partial \Omega_+$, $\Sigma
\subset \partial \Omega_0 \cap \partial \Omega_-$, and $\Sigma
\subset \partial \Omega_- \cap \partial \Omega_+$. The first two
situations are the cases of \emph{graphical blow-up} and
\emph{graphical blow-down} so there exists a solution of Jang's
equation $u_0 : \Omega_0 \to \mathbb{R}$ which diverges to positive,
respectively negative infinity on approach to $\Sigma \subset \partial
\Omega_0$. The strong maximum principle rules out the possibility that
$\Sigma$ be unstable in these cases straight away in view of the sub
and super solutions constructed in the preceding paragraph. The third
situation is the case of cylindrical blow-up: concretely, there is a
family of graphs $\{u_i\}_{i=1}^{\infty} \subset
\mathcal{C}^\infty(\Omega)$ where $u_i = u_{\tau_i}$ solve the
regularized Jang's equation $\mc(u_{\tau_i}) - \tr(k)(u_{\tau_i}) =
\tau_i u_{\tau_i}$ where $\tau_i \searrow 0$, and such that $u_i \to
\pm \infty$ uniformly on compact subsets of $\Omega_\pm$. From the
analysis of Schoen and Yau in \cite[Proposition
4]{Schoen-Yau:1981:PMTII} it follows that the hypersurfaces $\graph
(u_i) \subset \Omega \times \mathbb{R}$ converge smoothly on compact
sets to the `marginally trapped cylinder' $\Sigma \times
\mathbb{R}$. Note that $\graph(u_i)$ is a super solution of Jang's
equation where $u_i < 0$ and a sub solution where $u_i > 0$. Again
using vertical translates of the sub and super solutions for Jang's
equation above we can rule out the scenario that a component of a MOTS
$\Sigma$ arising in such a cylindrical blow-up be unstable.

The next case to deal with is the Plateau problem. A MOTS $\Sigma$
with boundary is stable in the sense of MOTS if there exists a
function $f$ on $\Sigma$, positive in the interior and vanishing on
the boundary, so that $\L_\Sigma f \geq 0$. We note here that the
Krein-Rutman theorem applies as before to show the existence of a real
eigenfunction (with Dirichlet boundary data) of $\L_\Sigma$ that is
positive on at least one component of $\Sigma$ and vanishing on all
the others.  For MOTS with boundary we have the following:
\begin{lemma}[\cite{Eichmair-Metzger:2010}]
  Assumptions as in Theorem \ref{thm:Plateau_problem} in
  the smooth dimensions $2 \leq n \leq 7$. Then there exist solutions
  $\Sigma$ of the Plateau problem for MOTS in $\Omega$ with boundary
  $\Gamma$ that are stable in the sense of MOTS.
\end{lemma}

Schoen and Yau used
\eqref{eq:rough_consequence_of_integrated_Schoen_Yau_identity} to
derive pointwise estimates for $|h|_\Sigma$ by adopting the iteration
technique of \cite{Schoen-Simon-Yau}. They obtained the area bounds
needed for this argument from a  calibration argument for solutions of Jang's
equation by comparison with extrinsic balls. These curvature estimates have been
generalized in \cite{Andersson-Metzger:2005} to stable MOTS. This
iteration technique extends to initial data sets $(M, g, k)$ of
dimension at most $6$. It is remarkable and important that---as with
minimal surfaces---in ambient dimension $3$, stable MOTS have curvature
estimates that are independent of a priori area bounds. This is uesed
crucially in section \ref{sec:existence-mots}.
\begin{theorem}[\cite{Andersson-Metzger:2005}]
  Let $\Omega \subset M$
  be a bounded open subset of an $n$-dimensional initial data set $(M,
  g, k)$ where $3 \leq n \leq 6$ and let $\Sigma \subset \Omega$ be a
  closed marginally outer trapped surface that is stable in the set of
  MOTS. Then one has the pointwise bound
  \begin{equation}
    |h|_\Sigma \leq C(\dist(\Sigma, \partial \Omega),
    |k|_{\mathcal{C}^1(\bar \Omega)}, |\Ric_M|_{\mathcal{C}(\bar \Omega)},
    \inj(\Omega, g), |\Sigma|).
  \end{equation}
  When $n=3$, then the bound
  on the right is independent of an a priori bound for the area
  $|\Sigma|$ of $\Sigma$.
\end{theorem}
We also mention that the regularity and compactness theory developed
in \cite{Schoen-Simon:1981} for stable minimal hypersurfaces
generalizes to \emph{embedded} MOTS, as was observed and used in
\cite{GAH}. This theory has the advantage of being available in all
dimensions provided one accepts the usual singular set of Hausdorff
co-dimension $7$. This furnishes a convenient framework to carry out
analysis on MOTS in high dimensions. This theory is particularly
effective when combined with a one-sided almost minimizing property,
see \cite[Appendix A]{GAH} and subsection \ref{sec:existence-mots}.

%%% Local Variables: 
%%% mode: latex
%%% TeX-master: "master"
%%% ispell-dictionary: en_US
%%% End: 

\section{Applications to general relativity}
\label{sec:applications}
In this section we discuss the applications that motivated the
development of the mathematical theory for Jang's equation.

\subsection{The positive mass theorem}
\label{sec:posit-mass-theor}
The first place where Jang's equation is analyzed is in its
application to reduce the general version of the positive mass theorem (PMT)
to its time-symmetric form due to Schoen and Yau in
\cite{Schoen-Yau:1981:PMTII}.

The positive mass theorem is a question about asymptotically flat
initial data sets $(M,g,k)$ and its ADM-mass and ADM-momentum.
\begin{theorem}[Positive mass theorem]
  \label{thm:PMT}
  If $(M,g,k)$ is a complete, asymptotically flat initial data set
  which satisfies the
nn  dominant energy condition $\mu\geq |J|$, then $m_\text{ADM}\geq
  |P_\text{ADM}|$.  Moreover, if $m_\text{ADM}=0$, then $(M,g,k)$ is
  initial data for Minkowski space.
\end{theorem}
In the maximal case, where $\tr_M k = 0$, the dominant energy condition implies
$\Scal_M\geq 0$. This leads to a formulation of the PMT relating only to the
Riemannian manifold $(M,g)$, called the Riemannian PMT.
\begin{theorem}
  \label{thm:PMT-Riem}
  Assume that $(M,g)$ is asymptotically flat and has $\Scal_M\geq
  0$. Then $m_\text{ADM}\geq 0$ and equality holds if and only if
  $(M,g)$ is flat $\IR^3$.
\end{theorem}
In a first step Schoen and Yau
\cite{Schoen-Yau:1979:PMTI,Schoen-Yau:1981:energy} proved the
Riemannian PMT in dimension 3 using the existence of certain area
minimizing slices. Their method extends to dimensions $3 \leq n \leq
7$ by a dimension reduction argument, see \cite{Schoen-Yau:1979} and
also \cite{Schoen:1987}. The minimal surface argument of Schoen and
Yau to prove the Riemannian PMT are closely related to their proof of
the non-existence of metrics of positive scalar curvature on the torus
in dimensions $n \leq 7$ in \cite{Schoen-Yau:1979-pos-scalar}. In
fact, Lohkamp observed in \cite{Lohkamp:1999} that the non-existence
of such metrics and the time symmetric positive mass theorem are
essentially equivalent in all dimensions. Two independent approaches
to extend the positive mass theorem to all dimensions by addressing
singularities of minimizing hypersurfaces when $n >7$ have been given
by Lohkamp \cite{Lohkamp:2006} and by Schoen.

An independent proof of Theorem~\ref{thm:PMT} using spinor
methods was later put forward by Witten
\cite{Witten:1981,Parker-Taubes:1982}. It does not need the
reduction of the PMT to the Riemannian PMT that we describe below, and works
in all dimensions under the topological assumption that the data set
be spin.

To describe the reduction of the general form of the positive mass
theorem to the Riemannian case using Jang's equation, we follow
\cite{Schoen-Yau:1979:PRL}. The actual argument due to Schoen and Yau
can be found in \cite{Schoen-Yau:1981:PMTII}. For the time being we
assume that $(M,g,k)$ is such that there exists a global solution $u$
to Jang's equation~\eqref{eq:jang} with boundary conditions $u\to 0$ at infinity. By
Theorem \ref{thm:existence_for_Jangs_equation} we know that such
solutions always exist provided $M$ does not contain any closed apparent
horizons.  The graph $\JM$ of $u$ with the induced metric $\Jg$ is
again asymptotically flat, and has the same ADM-mass as $(M,g,k)$. The
Schoen-Yau identity~\eqref{eqn:Schoen-Yau_identity} on $\JM$ implies, in
view of the dominant energy condition and a calculation similar to the
one in section~\ref{sec:stability_of_MOTS}, that for all functions
$\phi\in C^\infty(\JM)$ with compact support
\begin{equation}
  \label{eq:4}
  \int_\JM  2|D_\JM \phi|^2 + \phi^2 \Scal_\JM
  \geq
  \int_\JM |h-k|_\JM^2. 
\end{equation}
Written in a slightly different way this implies
\begin{equation}
  \label{eq:6}
  \int_\JM 8|D_\JM \phi|^2 + \phi^2 \Scal_\JM
  \geq
  6\int_\JM |D_\JM \phi|^2 
  \geq 0.
\end{equation}
It then follows from standard methods that there exists a positive
solution $\zeta$ of the equation
\begin{equation*}
  -\Delta_{\hat M}\zeta + \tfrac18\Scal_\JM\zeta
  =0,
\end{equation*}
such that $\zeta\to 1$ at infinity.
This
implies that the conformal metric $\tilde g := \zeta^4\Jg$ has
scalar curvature $\tilde\Scal_\JM = 0$. Moreover, it can be shown that
$\zeta$ has the asymptotic expansion
\begin{equation*}
  \zeta = 1 + A/r + O(r^{-2}).
\end{equation*}
Inserting $\zeta$ as test function into equation~(\ref{eq:6}) yields that
\begin{equation*}
  A \leq -\frac{6}{32\pi} \int_\JM |D_\JM \zeta|^2  \leq 0. 
\end{equation*}
That $\zeta$ is a legitimate test function can be verified by checking
that the boundary term in the integration by parts, used to
derive~\eqref{eq:6} from~\eqref{eqn:Schoen-Yau_identity} decays sufficiently fast.

A direct calculation shows that the ADM-mass of $(\JM,\tilde g)$
satisfies
\begin{equation*}
  m_\text{ADM}(\JM,\tilde g) = m_\text{ADM}(M,g) + \tfrac12 A \leq m_\text{ADM}(M,g)
\end{equation*}
so that the resulting manifold $(\JM,\tilde g)$ has ADM-mass no more
than the initial data set $(M,g,k)$. Since the scalar curvature is
zero, the Riemannian PMT gives that 
$m_\text{ADM}(M,g)\geq 0$.
If $m_\text{ADM}(M,g)=0$ one can work backwards through this argument
to see that in this case
$\tilde g$ is flat, $\zeta$ is constant, so that $\Jg =
\tilde g$. Moreover, there is also equality in the Schoen-Yau
identity, so that $h = k$. In particular,
the criterion of Jang~\eqref{eq:jang-thm1}
is satisfied and $(M,g,k)$ is a data set for Minkowski space.
 
Recall the simplifying assumption that a solution to Jang's equation
exists on $(M,g,k)$. This is indeed a restriction, as Jang's
equation can blow-up (or down) asymptotic to cylinders over marginally
outer (or inner) trapped surfaces, cf.\ Theorem \ref{thm:existence_for_Jangs_equation}.  The resolution of the situation was achieved in
\cite{Schoen-Yau:1981:PMTII} by compactifying the resulting
cylindrical ends using a conformal transformation. While the actual
procedure is out of the scope of this article, we wish to remark that
this step is a major obstacle in the reduction of the
general Penrose conjecture to the Riemannian version, proved by
Huisken and Ilmanen \cite{Huisken-Ilmanen:2001} and Bray
\cite{Bray:2001}. A detailed discussion of this fact can be found in
\cite{Malec-OMurchadha:2004}.

This reduction has been described by Schoen and Yau
\cite{Schoen-Yau:1981:PMTII} for three dimensional initial data
sets. The technical difficulties in higher dimensions are due to the
potential singularities of apparent horizons and hence the blow-up
cylinders in the solutions of Jang's equation, and their potentially
complicated topology that prevents direct application of the arguments
from \cite{Schoen-Yau:1981:PMTII}. These technical difficulties are
resolved in \cite{Eichmair:PMT} in dimensions $4\leq n\leq 7$.

\subsection{Formation of black holes}
\label{sec:condensation-matter}
The mechanism that causes Jang's equation to possibly blow-up along
apparent horizons yields an approach to the existence of apparent
horizons in the following way. Assume that $(M,g,k)$ is an initial
data set, where $M$ is compact with non-empty boundary. In addition
suppose that the boundary geometry is such that the barriers needed to
solve Jang's equation exist. Then the condition that $(M,g,k)$ does
not contain apparent horizons implies that the Dirichlet problem to
Jang's equation is solvable without the possibility of blow-up with
arbitrary boundary data, cf. Theorem
\ref{thm:existence_for_Jangs_equation}. If one can devise conditions
that lead to a contradiction using this solution, the existence of
apparent horizons can be concluded.

The first time that this prototype was used, is in the paper by Schoen
and Yau to prove the following theorem.
\begin{theorem}[\cite{Schoen-Yau:1983:cond}]
  \label{thm:condensation}
  Let $(M,g,k)$ be a compact initial data set with non-empty boundary
  $\del M$ such that $\mc_{\del M} > |\tr_{\del 
    M}(k)|$. Let $\Omega\subset M$ such that the following conditions are
  satisfied:
  \begin{enumerate}
  \item $\mu - |J| \geq \Lambda >0$ on $\Omega$,
  \item $\Rad(\Omega) \geq \sqrt{\frac{3}{2\Lambda}}\pi$.
  \end{enumerate}
  Then $M$ contains an apparent horizon.
\end{theorem}
Here $\Rad(\Omega)$ denotes the H-radius of a set $\Omega$ which is
defined as follows. Let $\Gamma\subset \Omega$ be a curve bounding a
disk in $\Omega$. The radius of $\Omega$ relative to $\Gamma$ is
defined as
\begin{equation*}
  \Rad(\Omega,\Gamma)
  :=
  \sup\{ r: \dist(\Gamma,\del\Omega)>r,\ \Gamma\ \text{does not bound a
    disk in}\ T_r(\Gamma)\}.
\end{equation*}
Here $T_r(\Gamma)$ is the tubular neighborhood of $\Gamma$ with radius
$r$. The radius of $\Omega$ then is defined as
\begin{equation*}
  \Rad(\Omega) = \sup \{ \Rad(\Omega,\Gamma) : \Gamma\ \text{bounds a
    disk in}\ \Omega\}.
\end{equation*}
Roughly speaking, $\Rad(\Omega)$ is the diameter of the largest tubular
neighborhood of a curve $\Gamma$ that does not contain a disk spanned
by $\Gamma$.

As already indicated the argument proceeds via contradiction, so
assume that there are no apparent horizons in $M$. Then the Dirichlet
problem for Jang's equation on $(M,g,k)$ is solvable with zero
boundary values, cf. Theorem \ref{thm:existence_for_Jangs_equation}. We denote the
graph of the solution by $\JM$. From
the Schoen-Yau identity~(\ref{eqn:Schoen-Yau_identity}), it
follows that on the portion $\JM_\Omega$ of $\JM$ above $\Omega$ one has
\begin{equation*}
  \Scal_\JM \geq 2\Lambda + 2 |\omega|^2 -2\div_\JM \omega,
\end{equation*}
where $\omega = X - D_\JM\log f$, $f=-\la\del_t,\Jnu\ra$, $\Jnu$ is
the downward unit normal to $\JM$, and $X$ is the
tangential part of the one form $-k(\Jnu,\cdot)$ as before.  This inequality yields
that the first Dirichlet eigenvalue $\lambda$ of the operator $\L:=
-\Delta_\JM + \tfrac12\Scal_\JM$ on $\JM_\Omega$ satisfies $\lambda \geq
\Lambda$. Furthermore, since the distances in the $\Jg$ metric are
no less than in the $g$ metric, it also follows that $\Rad(\JM_\Omega)
\geq \sqrt{\frac{3}{2\Lambda}}\pi$.

The first Dirichlet eigenfunction $\phi$ of $\L$ on $\JM_\Omega$, satisfying $-\Delta_\JM \phi
+ \tfrac12\Scal_\JM \phi = \lambda \phi$ is positive and can be used to define
the following functional on surfaces $\Sigma\subset \JM$,
\begin{equation*}
  A_\phi(\Sigma) = \int_\Sigma \phi \,d\hat\sigma,
\end{equation*}
where $d\hat\sigma$ denotes the area element induced by $\hat g$. Note,
that this functional can be interpreted as the area functional for
surfaces of the form $\Sigma\times S^1$ in $\JM\times S^1$ equipped
with the metric $\tilde g = \Jg + \phi^2 ds^2$, where $ds^2$ denotes the
standard metric on $S^1$. Note that $\tilde g$ has scalar curvature
$\tilde \Scal = \bar \Scal_\JM - 2 \phi^{-1}\Delta_\JM \phi = 2\lambda\geq 2
\Lambda$ in $\JM_\Omega$. This interpretation also implies that one
can find a minimizing disc $\Sigma$ for $A_\phi$ with boundary $\Gamma$,
where $\Gamma$ is chosen such that $\Rad(\JM_\Omega,\Gamma) \geq
\Rad(\JM_\Omega) - \varepsilon$. The minimizer $\Sigma$ satisfies the
Euler-Lagrange equation $\mc = -\la D_\JM\log \phi,\nu\ra$, where $\nu$
is the normal vector field on $\Sigma$ used to define $\mc$. More
importantly, stability of $\Sigma$ implies that the operator defined
by
\begin{equation*}
  f \mapsto -\Delta_\Sigma f - \la D_\Sigma\log \phi, D_\Sigma f\ra
  + f (\Lambda -\tfrac12 \Scal_\Sigma + \phi^{-1}\Delta_\Sigma \phi)
\end{equation*}
has non-negative Dirichlet spectrum. The form of this operator follows
for example by reduction of the stability operator in $(\JM\times
S^1,\tilde g)$ for equivariant variations on surfaces with
$S^1$-symmetry. Let $\psi>0$ be the first eigenfunction of this operator
and define a functional for curves $\gamma$ in $\Sigma$ as follows:
\begin{equation*}
  I_{\phi\psi}(\gamma)
  =
  \int_\gamma \phi\psi.
\end{equation*}
Recall Bonnet's theorem, which asserts that stable geodesics in surfaces
with scalar curvature bounded below by a positive constant have
bounded length. Here, the modification of the length functional by
introducing the weight $\phi\psi$ into $I _{\phi\psi}$ leads to a similar effect
for curves minimizing $I_{\phi\psi}$ in the sense that the stability of the
minimizer forces the minimizer to be short, in particular the length
is bounded by $\sqrt{\frac{3}{2\Lambda}}\pi$. By definition of
$\Rad(\JM,\Gamma)$ the disc $\Sigma$ intersects the boundary of a
tubular neighborhood $T_r(\Gamma)$ of radius $r<
\Rad(\JM,\Gamma)$. Thus it is always possible to find a minimizer for
$I_{\phi\psi}$ with length at least $r$, since one can minimize $I_{\phi\psi}$
among all curves with one endpoint on $\Gamma$ and one endpoint on
$\del T_r(\Gamma) \cap \Sigma$. This yields the desired estimate for
$\Rad(\JM)$, and thus for $\Rad(M)$. See \cite{Schoen-Yau:1983:cond}
for details.

There are several variations on this theme in the literature. Clarke
\cite{Clarke:1988} gave an interesting and useful observer independent
condition on the energy-momentum
tensor of a space-time that implies the trapping condition on the
boundary of the initial data set in Theorem~\ref{thm:condensation}.

A refined criterion for the existence of horizons was given by Yau
\cite{Yau:2001}.
\begin{theorem}
  Let $(M,g,k)$ be an initial data set satisfying the following
  conditions:
  \begin{enumerate}
  \item There exists $c>0$ such that $\mc_{\del M} - |\tr_{\del M}(k)| > c$.
  \item $\Rad(M) \geq \sqrt{\frac{3}{2\Lambda}}\pi$ where $\Lambda \leq
    \frac{2}{3}c^2 +\mu - |J|$ on $M$.
  \end{enumerate}
  Then M contains an apparent horizon. 
\end{theorem}
It is instructive to consider the case $k \equiv 0$ first: if $c>0$ is 
large, the first condition suggests that $M$ shrinks quickly from its 
boundary $\partial M$ inwards, while the second condition implies that
the interior of $M$ is large in a certain sense. The conclusion is that part of the 
interior of $M$ must be separated from the boundary by a minimal surface.

Again, Jang's equation enters prominently. Yau's argument in
\cite{Yau:2001} is by contradiction and proceeds as follows.  Assume
in virtue of Theorem~\ref{thm:existence_for_Jangs_equation} that
Jang's equation has a global solution $u$ on $M$. Denote the graph of
$u$ in $M\times\IR$ by $\JM$ and its induced scalar curvature by
$\Scal_\JM$. Then, by the Schoen-Yau
identity~(\ref{eqn:Schoen-Yau_identity}) one has that
\begin{equation*}
  2 (\mu -|J|) \leq \Scal_\JM - 2 |\omega|^2 + 2\div_\JM \omega,
\end{equation*}
where $\omega = X + D_\JM \log v$ as in
equation~(\ref{eqn:Schoen-Yau_identity}), where we use $v$ to denote $f^{-1}$.
This yields for all $\phi\in C^\infty(\JM)$ the following estimate:
\begin{equation}
  \label{eq:5}
  2\int_\JM (\mu - |J|)\phi^2 
  \leq
  \int_\JM 2|D_\JM \phi|^2 + \Scal_\JM \phi^2 
  +
  2 \int_{\del \JM} \phi^2 \la\omega,N\ra, 
\end{equation}
where $N$ denotes the outward pointing normal to $\del \JM$ in $\JM$. The
point is that the difference of the boundary term in
equation~(\ref{eq:5}) and the mean curvature of the boundary has a
positive lower bound, as one can see as follows. Recall that $
\la\omega,N\ra = \la D_\JM \log v,N\ra - k(\Jnu, N)$ where $v = f^{-1} =
\sqrt{1+|D_Mu|^2}$. Moreover, a calculation shows that the mean
curvature of $\del \JM$ in $\JM$ satisfies $\mc_{\del \JM} = v^{-1} \mc_{\del
  M}$, where the latter is calculated with respect to the metric
$g$. The normal $N$ is given by $N=v^{-1}(\eta + |D_Mu|\del_t)$, where
$\eta$ is the outward pointing normal to $\del M$ in $M$. To calculate
$\mc_{\del \JM}-\la\omega,N\ra$, note that since $u=0$ on $\del M$ we have that $D_Mu =
\sigma |D_Mu|\eta$, where $\sigma\in\{\pm 1\}$.  We let $V = \pi_* \Jnu
= v^{-1}D_Mu= \sigma v^{-1}|D_Mu|\eta$. Then the mean curvature of $\JM$
on $\del \JM$ is given by
\begin{equation*}
  \begin{split}
    \mc_\JM
    &=
    \div_M V
    =
    \sigma \div_M(v^{-1}|D_Mu| \eta)
    =
    \sigma v^{-1} |D_Mu| \mc_{\del M} + v^{-3}D_M^2u(\eta,\eta)
    \\
    &=
    \sigma v^{-1} |D_Mu| \mc_{\del M} + \sigma |D_Mu|^{-1} \la D_\JM\log v,N\ra.
  \end{split}
\end{equation*}
Note that this is also true if $D_Mu=0$, since then also $D_M\log
v=0$. By Jang's equation, $\mc_\JM = \tr_\JM(k)$, where
\begin{equation*}
  \tr_\JM(k)
  =
  \tr_M(k) - k(\Jnu,\Jnu)
  =
  \tr_{\del M}(k) + v^{-2} k(\eta,\eta).
\end{equation*}
Since furthermore $k(\Jnu,N) = \sigma v^{-2}|D_Mu| k(\eta,\eta)$ it
follows that
\begin{equation*}
  \begin{split}
    0
    &=
    \sigma |D_Mu| (\mc_\JM -\tr_\JM(k))\\
    &=
    v^{-1} |D_Mu|^2 \mc_{\del M} + \la D_M\log v,N\ra
    -
    \sigma \tr_{\del M}(k) - k(\Jnu,N)
  \end{split}
\end{equation*}
and thus
\begin{equation*}
  \la\omega,N\ra  = \sigma \tr_{\del M}(k) - |D_Mu|^2 \mc_{\del M}.
\end{equation*}
Finally, 
\begin{equation*}
  \mc_{\del \JM} - \la\omega,N\ra
  =
  v \mc_{\del M} - \sigma |D_Mu| \tr_{\del M} (k)
  \geq
  v(\mc_{\del M} - |\tr_{\del M}(k)|)
  \geq
  c.
\end{equation*}
This boundary term then has a similar effect as the $\Lambda$ in an
extension of the argument of Schoen and Yau to get an estimate on
$\Rad(M)$ contradicting the assumption as before.

Of the several different proposals to define the size of a body in an
alternative way, we want to mention specifically the suggestion of
Galloway and O'Murchadha \cite{Galloway-OMurchadha:2008}. They use the
intrinsic diameter of the largest stable MOTS bounded by curves in the
boundary of the body to define the radius of the body and show in turn
that this radius is bounded if the matter content of the body is
large.

The boundary effect discovered by Yau plays a crucial role in the
proof that the quasi-local mass defined by Liu and Yau is non-negative
\cite{liu-yau:2004,liu-yau:2006}. The common theme with
section~\ref{sec:posit-mass-theor} is that Jang's equation is used to
transform the question whether the Liu-Yau mass is non-negative to a
question in Riemannian geometry. As before the transition to Jang's
graph is followed by a conformal transformation to a metric with zero
scalar curvature. In the Riemannian setting established by this
procedure, the Liu-Yau mass is transformed to a quantity bounded below
by a modified version of the Brown-York mass. This uses the boundary
effect calculated above in a crucial way. Liu and Yau show that this
quantity is non-negative by extending an argument of Shi and Tam
\cite{shi-tam:2002}.

Eardley \cite{Eardley:1995} uses Jang's equation to give a different
criterion for the formation of black holes. To this end, for a data
set $(M,g,k)$ and a region $\Omega\subset M$ the following quantity is
introduced
\begin{equation*}
  k_\text{min}(\Omega) := \inf \{ \tr_M(k) - k(v,v) : p \in \Omega, v\in
  T_p M, |v|\leq 1\}.
\end{equation*}
Note that $k_\text{min}$ is the smallest value that $\tr k(u)$ could
take at any point of $\Omega$ for any graph $u:\Omega\to\IR$. In the
setting below, where $k$ is positive definite, it equals the minimal
value of the sum of two smallest eigenvalues of $k$ on $\Omega$.
\begin{theorem}[\cite{Eardley:1995}]
  \label{thm:eardley}
  Given compact initial data $(M,g,k)$ with non-empty boundary such
  that $\mc_{\del M}> |\tr_{\del M}(k)|$. If there is $\Omega\subset
  M$ such that 
  \begin{equation*}
    k_\text{min}(\Omega) \Vol(\Omega) > \Area(\del\Omega),
  \end{equation*}
  then there is an apparent horizon in $M$.
\end{theorem}
\begin{proof}
  The proof of this theorem is by contradiction. If there are no
  apparent horizons in $M$, then there exists a global solution of Jang's
  equation to the Dirichlet problem with zero boundary data. Denote the
  graph of this solution by $\JM$ and by $V = \pi_* \Jnu$, the
  orthogonal projection 
  of the downward unit normal. Then Jang's equation is equivalent to
  \begin{equation*}
    \div_M V = \tr_M(k) - k(V,V),
  \end{equation*}
  since $\div_M V$ is the mean curvature of $\JM$ with respect to
  $\Jnu$ and the right hand side is just the trace of $k$ on
  $\JM$. Integrating this on $\Omega\subset M$ yields
  \begin{equation*}
    k_\text{min}(\Omega)\Vol(\Omega)
    \leq
    \int_\Omega \tr_M(k) - k(V,V) 
    =
    \int_{\del\Omega} \la V,\nu\ra 
    \leq
    \Area(\del\Omega).
  \end{equation*}
  Here $\nu$ denotes the outward normal to $\del\Omega$ in $M$. This
  contradicts the assumptions of the theorem.
\end{proof}

\subsection{Existence and properties of outermost MOTS}
\label{sec:existence-mots}

In section~\ref{sec:posit-mass-theor} the potential blow-up of Jang's
equation at apparent horizons is an undesirable property that has to
be overcome. In section~\ref{sec:condensation-matter} the existence of
apparent horizons is a rather indirect consequence. In contrast, the
way Jang's equation is used to construct MOTS in section~\ref{sec:existence_of_MOTS_due_to_blow_up} is far more direct and can be used
to derive crucial properties of outermost MOTS.

To get started, fix a complete initial data set $(M,g,k)$, and assume
for simplicity that $M$ is compact and that $\del M$ satisfies
$\theta^+_{\del M}> 0$.  We say that a MOTS $\Sigma\subset M$ is
\emph{outermost} if it is of the form $\Sigma=\del \Omega$, where
$\Omega\subset M$, and the following holds: If $\Sigma'=\del\Omega'$
is any other MOTS, with $\Omega'\supset \Omega$, then
$\Omega'=\Omega$. In other words, if $\Sigma$ is outermost, then there
is no MOTS in the region $M\setminus\Omega$ exterior to $\Sigma$.

We expect the outermost MOTS to be the boundary of the trapped
region. To this end, we define a set $\Omega\subset M$ to be
\emph{trapped}, if $\theta^+_{\del\Omega}\leq0$. The trapped region
$\CT$ is then the union of all trapped sets \cite{Wald,HawEll},
\begin{equation*}
  \CT = \bigcup_{\Omega\ \text{is trapped}} \Omega.
\end{equation*} 
Using a slight extension of the existence Theorem
~\ref{thm:Schoens_theorem} adapted to weakly trapped boundaries,
cf. \cite[Section 5]{Andersson-Metzger:2009} and also \cite[Remark
4.1]{GAH}, it follows that a trapped region $\Omega$ as in the
definition of $\CT$ is contained in a trapped region $\Omega' \supset
\Omega$ whose boundary $\partial \Omega'$ is a MOTS, and such that
$\partial \Omega'$ is stable in the sense of MOTS and is $C$-almost
minimizing with respect to variations in $M \setminus \Omega'$.

To conclude smoothness of $\del\CT$ as for example in
\cite{Huisken-Ilmanen:2001} where the time-symmetric case $k \equiv 0$
is discussed, we need to verify three points. These are whether two
intersecting MOTS are contained inside one smooth MOTS that encloses
them, the embeddedness of $\del\CT$, and area bounds.

The question whether the union of two trapped sets is a trapped set
relates to the following problem. Given a sequence of MOTS, $\Sigma_n
= \del \Omega_n$, we wish to replace it by an increasing sequence
$\Sigma_n' = \del \Omega_n'$ so that $\Omega_m' \subset
\Omega_n'$ for all $m \leq n$, as in \cite{Huisken-Ilmanen:2001}. This
can be handled in two different ways. In \cite{Andersson-Metzger:2009}
a sewing lemma due to Kriele and Hayward \cite{Kriele-Hayward:1997}
was employed in conjunction with Theorem~\ref{thm:Schoens_theorem}
to conclude that if two sets $\Omega_1$ and $\Omega_2$ with
$\theta^+_{\del\Omega_i}= 0$ intersect, then there is $\tilde
\Omega\supset \Omega_1\cup\Omega_2$ with $\theta^+_{\del\tilde\Omega}
=0$. Alternatively, the Perron method and an
approximation argument can be used to find an enclosing MOTS
\cite[Remark 4.1]{GAH}.

To conclude embeddedness of $\del\CT$ we have to show that the limit
of such an increasing sequence of MOTS $\Sigma_n=\del\Omega_n$ is
embedded. Since all the $\Sigma_n$ are increasing and embedded, the
only crucial point is that the limit $\Sigma$ does not touch
itself on the outside. For minimal surfaces this scenario would be
ruled out by the maximum principle, which does not work for MOTS in
this situation. The problem is that locally two sheets of a MOTS may
touch, but with opposite orientation. The case of two touching spheres
in flat space illustrates this.  To show that this can be ruled out
for outermost MOTS, in \cite{Andersson-Metzger:2009} a quantity called
the \emph{outward injectivity radius} was introduced. It is then shown
that one can assume it to be bounded below along the sequence
$\Sigma_n$ as above. This bound yields a lower bound on the arc length
of a geodesic starting on $\Sigma_n$ in direction of the outer normal,
before it can intersect $\Sigma_n$ again. The argument in
\cite{Andersson-Metzger:2009} derives this property from the fact that
along a short geodesic that joins two points on $\Sigma_n$, a neck
with negative $\theta^+$ can be inserted. Then the sewing lemma can be
used to produce a barrier suitable for Theorem~\ref{thm:Schoens_theorem}. This procedure can only be applied a finite number of
times, since it can be shown that each surgery can be made at a place
where it consumes a fixed amount of volume outside of the initial
MOTS. This surgery requires curvature bounds for stable MOTS, which
have been derived in \cite{Andersson-Metzger:2005} in ambient
dimension 3.  Alternatively, one can use results from
\cite{GAH} based on the lower order properties of horizons
and the regularity theory of Schoen-Simon to conclude embeddedness of
$\Sigma$.  In fact, it is easy to see that if two sheets of the
hypersurface $\Sigma_n$ are close on the outside as above, then they
can be joined by a small catenoidal neck to save area. This would
contradict the almost minimizing property of $\Sigma_n$ with respect
to variations in the complement of $\Omega_n$. This approach also
works in higher dimensions.

This leaves as a last point the fact that the area of the surfaces
$\Sigma_n$ needs to be bounded. In \cite{GAH} these bounds are
immediate from the almost minimizing property. In
\cite{Andersson-Metzger:2009} it is shown that a lower bound on the
outward injectivity radius implies an upper bound on the area. This
follows from the observation that, given curvature estimates, the area
of the MOTS can be estimated by the volume of the outward part of an
embedded tubular neighborhood of radius $\rho$ divided by
$\rho$. Since the outward injectivity radius is bounded below for
outermost MOTS, one can take a fixed $\rho$ and conclude the area
bounds from the fact that there is only finite volume outside the
MOTS. The approach in \cite{Andersson-Metzger:2009} is specific to
ambient dimension three, since the lower bound on the outward
injectivity radius requires the surface in question to have curvature
bounded independently of the area.

Let us investigate the topology of the outermost MOTS. In three
dimensions outermost MOTS (assuming an outer untrapped barrier) are
unions of topological spheres. This is well known in the time
symmetric case in three dimensions, where MOTS are minimal surfaces,
for example \cite{Galloway:1993} uses minimal surfaces techniques from
\cite{Meeks-Simon-Yau:1982} or \cite[Lemma 4.1]{Huisken-Ilmanen:2001}
where this is proven without curvature restriction.

For MOTS the question of topology was answered by Galloway and Schoen
\cite{Galloway-Schoen:2006}, who showed that any stable MOTS must be
of non-negative Yamabe type, provided the dominant energy condition
holds. The argument is based on the Schoen-Yau identity, which follows
for stable MOTS and a calculation similar to
section~\ref{sec:posit-mass-theor}. Galloway \cite{Galloway:2008} was
able to exclude the marginal case for smooth outermost MOTS. The
argument is based on an observation that in case of $\tr k \leq 0$ a
stable MOTS $\Sigma$ with Yamabe type 0 has an integrable Jacobi field
that leads to a local foliation by MOTS on the outside of $\Sigma$,
which contradicts the condition that $\Sigma$ be outermost. The case
of a general $\tr k$ can be reduced to this case by bending the data
$(M,g,k)$ in its ambient space-time to the past, and using the
Raychaudhuri equation to show that the foliation of MOTS in this new
slice gives rise to trapped surfaces outside of $\Sigma$ in the
original data set.

Collecting these results, we arrive at the following comprehensive
theorem about the existence, regularity, and properties of the trapped
region
\cite{Andersson-Metzger:2005,Andersson-Metzger:2009,GAH,Plateau,Galloway-Schoen:2006,Galloway:2008}.
\begin{theorem}
  Assume that $(M,g,k)$ is an asymptotically flat initial data set of
  dimension $2 \leq n\leq 7$. There is an explicit constant $C>0$
  depending only on the geometry of $(M,g,k)$ such that the following
  hold:

  If the trapped region $\CT$ of $(M, g, k)$ is non-empty, then
  $\del\CT$ is a smooth, embedded, outermost and stable MOTS. The area
  and the second fundamental form of $\del \CT$ are bounded by $C$ and its
  outward injectivity radius is bounded below by
  $\frac{1}{C}$. Furthermore, $\del\CT$ is $C$-almost minimizing with
  respect to variations in $M \setminus \CT$.
  
  If $(M,g,k)$ satisfies the dominant energy condition, then $\del\CT$
  is the union of components with non-negative Yamabe-type. If
  $(M,g,k)$ is a slice of a space-time satisfying the dominant energy
  condition, then the components of $\del \CT$ have positive
  Yamabe-type.
\end{theorem}
For the explicit dependence of the constant, see the original
references \cite{Andersson-Metzger:2005,Andersson-Metzger:2009,GAH,Plateau}.

To conclude, we wish to point out that the existence of the trapped
region in $(M,g,k)$ allows the construction of blow-up solutions to
Jang's equation. These are nontrivial solutions to Jang's equation
which are defined on $M\setminus (\CT\cup \Omega_-)$ where
$\Omega_-\subset M$ is such that the boundary components of $\Omega_-$
disjoint from $\del\CT$ are MITS. The construction uses the techniques
discussed in section~\ref{sec:existence_of_MOTS_due_to_blow_up} and
is described in \cite{Metzger:2010}. A catch however is that some or
all of the components of $\del\CT$ may lie in the interior of
$\Omega_-$ if they are enclosed by surfaces $\Sigma$ satisfying $\mc_\Sigma -
\tr_\Sigma(k) =0$.
%%% Local Variables: 
%%% mode: latex
%%% TeX-master: "master"
%%% ispell-local-dictionary: "en_US"
%%% End: 

\section{Outlook}\label{sec:outlook}
In this section we indicate a couple of directions for further research 
related to the ideas discussed
in this survey. 
\subsection{Generalizations of Jang's equation} 
The Penrose inequality 
\begin{equation*}
  m_{\text{ADM}} \geq \sqrt{\frac{A(\Sigma)}{16\pi}}  
\end{equation*}
is an equality only for slices in the Schwarzschild spacetime. As we have seen, Jang's
equation was motivated by the idea of ``detecting'' data sets which generate
a Minkowski geometry $-dt^2 + g^{\text{flat}}$. Based on this observation, it
appears reasonable that 
any approach to proving the general Penrose inequality must utilize a setting
which is sensitive to the 
Schwarzschild geometry. 
Motivated by this line of thought, 
Bray and Khuri \cite{Bray-Khuri:2009A,Bray-Khuri:2009B}
recently 
extended Jang's equation to a system of equations which is designed to identify 
slices of the Schwarzschild space-time. 

Recall that the Schwarzschild spacetime in isotropic coordinates 
can be written as a warped product 
\newcommand{\Schw}{\text{Schw}}
with metric 
$
g^{\Schw} -dt^2 \phi^2 
$
where 
\begin{equation*}
  \phi = \frac{1-\frac{2m}{|x|}}{1+\frac{2m}{|x|}} 
  , \quad g^{\Schw} = \left ( 1 + \frac{m}{2|x|} \right) ^4 \delta_{ij} dx^i dx^j.
\end{equation*}
The condition that an initial data set $(M,g,k)$ can be represented as
a graph $(x, u(x))$ in the Schwarzschild spacetime is then that
\begin{equation*}
  g_{ab} = g^{\Schw}_{ab} - \phi^2 D_a u D_b u , \quad
  k_{ab} = \pi_{ab},
\end{equation*}
where $\pi_{ab}$ is the second fundamental form of $\graph{u}$ in the
Schwarzschild spacetime. As shown by Bray and Khuri, one may also in this
more general situation introduce defects in terms of which the
condition that $(M,g,k)$ is the data induced on $\graph(u)$ in the
Schwarzschild spacetime can be characterized. 
As in the classical setup, these data can be calculated in terms of a 
related Riemannian spacetime, which is a product
over $(M,g)$. In the generalization this is a
warped product over $M$ with warping function $\phi^2$, i.e.  
$(M \times \mathbb R, g + \phi^2 dt^2)$. This spacetime is additionally 
endowed with a symmetric
2-tensor $K$ which is a lift of $k$, the second fundamental from of $M$ in
the spacetime, to the warped product. Recall that in the classical Jang
equation, the lift of $k$ is simply $\pi^* k$, where $\pi$ is the vertical
projection. For the generalized Jang's equation, the lift $K$ is defined as 
\begin{equation*}
  K = \pi^* k + \phi d\phi(N)  dt^2 
\end{equation*}
where $N$ is the downward pointing normal of $\graph(u)$ in the warped product.

The generalized Jang's equation now takes the form 
\begin{equation}\label{eq:Jang-gen} 
\mc_{\JM} - \tr_{\JM} K = 0,
\end{equation} 
cf. \cite[section 2]{Bray-Khuri:2009B}. 
Due to the lack of symmetry in the warped product, it is necessary to
consider the warping function $\phi$ as an unknown and add an equation for
this as well. 

Bray and Khuri \cite{Bray-Khuri:2009A,Bray-Khuri:2009B} propose three
different systems of equations 
incorporating the generalized Jang's equation together with
equations for $\phi$, which have the potential for yielding a proof of a
Penrose inequality. As shown by the counter-example of Carrasco and Mars
\cite{carrasco:mars}, one version of the Penrose inequality proposed by Bray
and Khuri, in terms of generalized apparent horizons, is not
valid. However, in spite of this counter-example, the approach introduced by
Bray and Khuri may still be applicable to other versions of the 
Penrose inequality, see the survey paper \cite{mars:survey:PI} for further
discussion. 

The analysis of the systems proposed by Bray and Khuri is made more difficult
by the fact that $\phi$ tends to zero at the horizon and as a consequence the
generalized Jang's equation is degenerate there. 
Bray and Khuri have been able to carry out the necessary analysis in the
spherically symmetric case, providing a new proof of the general PI in this
restricted case. 

\newcommand{\HH}{\mathcal H} 

\subsection{Evolution of MOTS}
Consider a spacetime which is the
maximal development of asymptotically flat data on $M$ 
for an Einstein-matter system
satisfying the DEC. 
Supposing that the Cauchy surface contains a stable MOTS $\Sigma$, which we
without loss of generality can assume to be outermost, the spacetime
contains a black hole, and under some weak genericity conditions
the MOTS lies on a spacelike marginally outer
trapped tube (MOTT) $\HH$. The MOTT is determined by a choice of 
Cauchy foliation of the spacetime. 

The outermost MOTT is, with the exception of jump-times
(see below) space-like in the generic case. 
Thus, the MOTT is an outflow
boundary for causal equations in its exterior, and the maximal development of
the restriction of the 
Cauchy data on $M$ contains the exterior to $\HH$. 

\begin{figure}[!hbt]
\centering 
\resizebox{!}{2in}{\input{pic/outflow}} 
\caption{}
\end{figure}

This leads to an
exterior Cauchy problem for eg. the Einstein equations in spacetime harmonic coordinates. Let a Cauchy
surface $M$ be given, containing an outermost MOTS. The exterior Cauchy
problem is the initial-boundary value problem for the evolution of this
system in the closed exterior of the MOTT, including the MOTS boundary, 
evolving from the outermost MOTS. 
This problem can be expected to be relevant for the problem of Kerr
stability, and in particular it is interesting to prove a useful continuation
criterion for it.

If
we consider the maximal extension of the MOTT to the future in a spacetime
with a regular Cauchy foliation, one expects that
after a finite sequence of jumps \cite{Andersson-Metzger:2009}
this eventually approaches the event horizon. 
It is an interesting question to understand the details of this scenario. In
particular, in terms of the Kerr stability problem, one expects to have Price
law decay of the matter and gravitational energy flux across the event
horizon. It is reasonable to speculate that the corresponding statement holds
for the fluxes across the (weakly) spacelike MOTT. 
As the strength of the flux decreases this has the effect of turning the MOTT 
null.

This leads to the expectation that
the MOTT asymptotically approaches the event horizon and terminates at future
timelike infinity. Since the  MOTT is expected to rapidly turn null, one
expects the distance along the MOTT to its boundary at future timelike
infinity to be finite. This behavior was verified in the spherically
symmetric case by Williams
\cite{williams:2008} who showed that for an Einstein-scalar field spacetime
with decay along the event horizon of the form $v^{-2-\eps}$, the MOTT has
its boundary at a finite distance. As pointed out by Williams, the required 
decay is
weaker than the expected Price law decay of $v^{-3}$. If this scenario is
correct, it is likely 
there is a relation between the decay of fluxes across the MOTT and the
regularity at the boundary of the MOTT at future timelike infinity. 
We mention here also the work of Ashtekar and Krishnan
\cite{ashtekar:krishnan:2002,ashtekar:krishnan:2003} in
the dynamical horizon (DH) 
setting\footnote{A MOTT is a dynamical horizon if it is
spacelike and foliated by marginally trapped surfaces, i.e. MOTS which also have
negative expansion with respect to the ingoing null normal, see
\cite[section 2.2]{ashtekar:galloway:2005} for details.} 
\providecommand{\EH}{\text{event horizon}}
\providecommand{\Scri}{\mathcal{I}}
\begin{figure}[!hbt]
\centering 
\resizebox{!}{2in}{\input{pic/iplus}} 
\caption{}
\end{figure}
which shows that the
area of the cross sections of a DH is increasing (a quasi-local 
version of Hawking's area law for event horizons), and gives expressions for
the rate of increase of area in terms of the flux across the DH.

As discussed in section \ref{sec:existence-mots}, once a MOTS is created in
an evolving spacetime on a Cauchy surface $M_0$
then, if the spacetime satisfies the null
energy condition, each Cauchy slice in the future of $M_0$ contains
an outermost MOTS. Further, each time a MOTS $\Sigma_0$ 
is created, it is through a
bifurcation process which leads to an inner and an outer branch of the MOTT
originating at $\Sigma_0$. The outer
branch may jump but remains stable, while one expects that the inner branch
eventually becomes unstable. 

It is of interest to understand in more detail the space-time track
of the MOTS. 
The generalized maximum principle for MOTS, cf. section
\ref{sec:existence-mots} 
implies that two locally outermost MOTS which approach sufficiently closely
must eventually  be surrounded by a MOTS. 
In terms of the evolution of binary
black hole data this means that two black holes (as determined by their
apparent horizons) which approach sufficiently closely, eventually are
swallowed by a larger black hole surrounding the two. 

Ashtekar and Galloway
\cite{ashtekar:galloway:2005} proved a
uniqueness result which gives further information on the spacetime geometry
of dynamical horizons, a special case of MOTTs. This result states that in a
spacetime satisfying the null energy condition, 
the past domain of dependence of a DH
cannot contain a marginally trapped surface, see \cite[Theorem
  4.1]{ashtekar:galloway:2005}. It would be interesting to understand better
whether results of this type hold for MOTTs and MOTSs. 

If one considers two BH's, one of which is small relative to the other, it is
natural to consider a scenario where the small BH falls into the larger
one. In this case, the generalized maximum principle for MOTS does not give
any information about the small BH crossing the horizon of the large one, but
the classical maximum principle prevents one MOTS from ``sliding'' inside
another. In particular, the configuration shown in fig. \ref{fig:twoBH},
corresponding to the moment when the small BH moves inside the larger BH is
ruled out by the maximum principle. 
\renewcommand{\EH}{\text{event horizon}}
\begin{figure}[!hbt]
\centering 
\resizebox{!}{2in}{\input{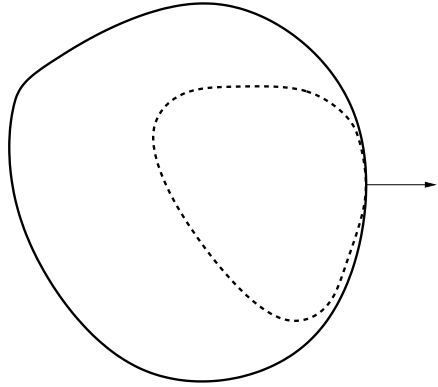}} 
\label{fig:twoBH}
\caption{This situation is ruled out by the maximum principle}
\end{figure}
Therefore one expects that as the BH's coalesce, the two apparent horizons
will eventually approach each other and merge. It is interesting to speculate
whether the MOTS in such a situation form a
continuous spacetime track, with one branch connecting the merging horizons
with the outermost, surrounding, MOTT. See \cite{Jeff} for details. 

\section{Concluding remarks} 
In this paper we have given a survey of the state of the art concerning
Jang's equation, MOTS, implications on the existence of black holes and related issues. 
The main motivation for considering
these issues has so far been in the asymptotically flat case. However, it is
important to recall that also in considering the Cauchy problem for the
Einstein equations in strong field situations, analogues of MOTS and trapped
regions can be expected to play an important role, and therefore some of the
topics discussed in this survey may have applications in future
approaches to global evolution problems and the cosmic censorship problem.

\subsection*{Acknowledgements} 
LA and ME are grateful to the organizers of CADS IV for their support
and hospitality during the conference in Nahariya.  We wish
to thank Robert Bartnik, Hubert Bray, Markus Khuri, Marc Mars, Pengzi Miao, Todd
Oliynyk, Richard Schoen, and Walter Simon for helpful conversations on
topics related to Jang's equation and the Penrose inequality.

%%% Local Variables: 
%%% mode: latex
%%% TeX-master: "master"
%%% ispell-dictionary: en_US
%%% End: 

\bibliographystyle{abbrv}
\bibliography{references}
\end{document}